\documentclass[a4paper,USenglish,numberwithinsect,cleveref,autoref,thm-restate]{lipics-v2021}

\usepackage{tikz}
\usetikzlibrary{decorations.shapes,decorations.pathreplacing,decorations.pathmorphing}
\usetikzlibrary{arrows,matrix,shapes}
\usetikzlibrary{positioning}
\tikzset{
        stars/.style={star,inner sep=2pt}
    }

\usepackage{todonotes}
\usepackage[ruled,vlined,linesnumbered]{algorithm2e}
\usepackage{hyperref}
\usepackage{xcolor}
\usepackage{cleveref}

\usepackage{microtype}

\newtheorem{reduction rule}{Reduction Rule}[section]

\newcommand{\cA}{{\mathcal A}}
\newcommand{\cB}{{\mathcal B}}

\newcommand{\cT}{{\mathcal T}}
\newcommand{\ccC}{{\mathcal C}}

\newcommand{\OO}{{\mathcal O}}

\newcommand{\XX}{{\mathcal X}}

\newcommand{\bN}{{\mathbb N}}

\newcommand{\dptw}{{\sf dp}}

\newcommand{\hind}{{\textsf{h-index}}}

\newcommand{\sclub}{{\sc $s$-Club}}
\newcommand{\twoclub}{{\sc $2$-Club}}
\newcommand{\clique}{{\sc Clique}}

\newcommand{\trangle}{{\sc Vertex $r$-Triangle}}
\newcommand{\etrangle}{{\sc Edge $r$-Triangle}}
\newcommand{\vonetriangle}{{\sc Vertex 1-Triangle}}

\newcommand{\rVT}{{\sc $r$-VT}}

\newcommand{\rET}{{\sc $r$-ET}}

\newcommand{\trw}{{\sf tw}}
\newcommand{\fes}{{\sf fes}}
\newcommand{\vc}{{\sf vc}}

\newcommand{\fvs}{{\sf fvs}}
\newcommand{\mdeg}{{\sf maxdeg}}
\newcommand{\bnw}{{\sf bw}}

\newcommand{\adjc}{{\rm adj}}
\newcommand{\trn}{{\rm trn}}

\newcommand{\vtrangle}{{vertex $r$-triangle property}}
\newcommand{\egtrangle}{{edge $r$-triangle property}}

\newcommand{\nka}{\textrm{$NP \subseteq coNP/poly$}}

\newcommand{\defproblem}[3]{
  \vspace{3mm}
\noindent\fbox{
  \begin{minipage}{.95\textwidth}
  \begin{tabular*}{\textwidth}{@{\extracolsep{\fill}}lr} \textsc{#1} 
   \\ \end{tabular*}
  {\bf{Input:}} #2  \\
  {\bf{Question:}} #3
  \end{minipage}
  }
  \vspace{2mm}
}

\newtheorem{nclaim}{Subclaim}[section]

\title{On the Structural Parameterizations of Finding 2-Clubs with Triangle Constraints} 

\titlerunning{Structural Parameterizations of 2-Clubs with Triangle Constraints} 

\author{Ashwin Jacob}{National Institute of Technology Calicut, India}{ashwinjacob@nitc.ac.in}{https://orcid.org/0000-0003-4864-043X}{}

\author{Diptapriyo Majumdar}{Indraprastha Institute of Information Technology Delhi, New Delhi, India}{diptapriyo@iiitd.ac.in}{https://orcid.org/0000-0003-2677-4648}{}

\author{Raghav Sakhuja}{Indraprastha Institute of Information Technology Delhi, New Delhi, India}{raghav21274@iiitd.ac.in}{}{}

\authorrunning{Jacob et al.} 

\Copyright{Jane Open Access and Joan R. Public} 



\category{} 

\relatedversion{} 

\nolinenumbers 

\EventEditors{John Q. Open and Joan R. Access}
\EventNoEds{2}
\EventLongTitle{42nd Conference on Very Important Topics (CVIT 2016)}
\EventShortTitle{CVIT 2016}
\EventAcronym{CVIT}
\EventYear{2016}
\EventDate{December 24--27, 2016}
\EventLocation{Little Whinging, United Kingdom}
\EventLogo{}
\SeriesVolume{42}
\ArticleNo{23}

\begin{document}

\maketitle


\begin{abstract}
Given an undirected graph $G = (V, E)$ and an integer $k$, the
{\sclub} asks if $G$ contains a vertex subset $S$ of at least $k$ vertices such that $G[S]$ has diameter at most $s$.
Recently, {\trangle} {\sclub}, and {\etrangle} {\sclub} that generalize the notion of {\sclub} have been studied by Garvardt et al. [TOCS-2023, IWOCA-2022] from the perspective of parameterized complexity.
Given a graph $G$ and an integer $k$, the {\trangle} {\sclub} asks if there is an {\sclub} $S$ with at least $k$ vertices such that every vertex $u \in S$ is part of at least $r$ triangles in $G[S]$.
In this paper, we initiate a systematic study of {\trangle} {\twoclub} for every integer $r \geq 1$ from the perspective of structural parameters of the input graph.
In particular, we provide an FPT algorithm for {\trangle} {\twoclub} when parameterized by the treewidth (${\trw}$) of the input graph and corresponding lower bound.
Additionally, when parameterized by the feedback edge number (${\fes}$) of the input graph, 
we provide a kernel of $\OO({\fes})$ edges for {\trangle} {\twoclub}.
\end{abstract}



 
\section{Introduction}
\label{sec:intro}


Given an undirected graph $G = (V, E)$, a {\em clique} is a subset $S \subseteq V(G)$ such that every pair of vertices of $S$ are adjacent.
Given a graph $G = (V, E)$ and an integer $\ell$, {\clique} problem asks if $G$ has a clique having at least $\ell$ vertices.
{\clique} is a well-studied problem in the domain of algorithms and complexity theory, and is known to be W[1]-Complete from the perspective of parameterized complexity.
Observe that a clique is a subgraph of a graph that has diameter at most one.
A natural relaxation of {\clique} is {\sclub} that asks to find a set $S$ of at least $\ell$ vertices that induces a subgraph of diameter at most $s$.
This relaxation of $s$-club was proposed by Mokken \cite{Mokken16}.
%
Formally, the input to {\sclub} is an undirected graph $G$ and an integer $\ell$, and the question asked is whether $G$ has a set $S$ of at least $\ell$ vertices such that the diameter of $G[S]$ is at most $s$.
If $s=1$, then {\sclub} is precisely the same as {\clique} problem, and if $s = 2$, then {\sclub} is the same as {\twoclub} problem.
Both {\twoclub} and {\sc $3$-Club} have been found as natural relaxations of {\clique}.
In particular, there have been numerous works both from the perspective of theory, as well as from experiments, that concern finding 2-clubs and 3-clubs of a graph (see \cite{BalasundaramBT05,Balasundaram09,HartungKN15,HartungKNS15,PajouhB12,PajouhBH16}). 

Later, this model of $s$-club was further augmented by several various other constraints \cite{AlmeidaB2019,Kanawati2014,PattilloYB2013} which asks to find an $s$-club but it also satisfies some additional properties.
One such model (see \cite{AlmeidaB2019}) asks that every vertex (respectively, every edge) of the $s$-club $S$ has to be part of at least $r$ triangles in $G[S]$.
Formally, a vertex subset $S \subseteq V(G)$ fulfills/satisfies {\em {\vtrangle}} (respectively, the {\em {\egtrangle}}) if every $v \in S$ (respectively, every edge of $G[S]$) is part of at least $r$ triangles in $G[S]$.
The following problem asks to find an $s$-club with at least $\ell$ vertices that satisfies the {\vtrangle}.

\defproblem{{\trangle} {\sclub} ({\rVT} {\sclub})}{An undirected graph $G = (V, E)$ and an integer $\ell$.}{Is there $S \subseteq V(G)$ with at least $\ell$ vertices such that $S$ is an $s$-club and $S$ satisfies {\vtrangle}?}

Similarly, {\etrangle} {\sclub} ({\rET} {\sclub}) asks to find an $s$-club $S$ with at least $\ell$ vertices that satisfies {\egtrangle}.
The $s$-club is one such model that captures cohesive substructures in social and biological networks, as they capture groups of vertices with small pairwise distances.
But plain $s$-clubs (even $2$-clubs) can produce fragile communities that are held together by only a few edges.
Augmenting $s$-clubs with triangle constraints eliminates such weaknesses and yields subgraphs that are of diameter $s$ and are robustly interconnected.
%
From the perspective of parameterized complexity, natural parameters are the {\em solution size} ($\ell$) and {\em number of triangles} ($r$) that arrive from the formulation of the problem itself.
Garvardt et al. \cite{GarvardtKS23} initiated the study of {\vtrangle} {\twoclub} and {\etrangle} {\twoclub} from the perspective of parameterized complexity.
The authors in their paper \cite{GarvardtKS23} have proved that if solution size is considered as the parameter, then the inclusion of {\vtrangle} requirement makes {\twoclub} W[1]-hard.
To the best of our knowledge, {\vtrangle} {\sclub} and {\etrangle} {\sclub} are unexplored with respect to the structural parameters that are described by the structural properties of the input graph.
Hence, it is interesting to explore these two problems with respect to the {\em structural parameters}.




\subparagraph*{Our Results.} Being motivated by the results in \cite{GarvardtKS23} and the study of {\twoclub} under structural parameters by \cite{HartungKN15} and \cite{HartungKNS15}, we initiate a systematic study of {\trangle} {\twoclub} with respect to structural parameters of the input graph in this paper.
In particular, our objective is to develop a concrete understanding if the tractability (or hardness) results of {\twoclub} transfer to tractability (or hardness) results to {\trangle} {\twoclub} when structural parameters are considered.
The structural parameters considered in this paper are the treewidth (${\trw}$), the size of a minimum feedback edge set (${\fes}$), and the ${\hind}$ of the input graph. 
We refer to Section \ref{sec:prelims} for definitions about various structural parameters.
There are two motivations behind the choice of these parameters.
First, the parameterized complexity status of {\trangle} {\twoclub} is largely unexplored with respect to structural parameters to the best of our knowledge.
Second, it is worth interesting to determine if the FPT (or hardness or polynomial kernelization) status of {\twoclub} transfer to {\trangle} {\twoclub} for various structural parameters.



Our first result is an FPT algorithm for {\trangle} {\twoclub} when parameterized by the treewidth (${\trw}$) of the input graph.
We prove the following result for {\trangle} {\twoclub} in Section \ref{sec:FPT-treewidth}.

\begin{restatable}{theorem}{ThmTwDPTriangle}
\label{thm:triangle-tw-dp}
For every $r \geq 1$, {\trangle} {\twoclub} parameterized by ${\trw}$, the treewidth of the input graph admits a $2^{2^{\OO({\trw})}}n^{\OO(1)}$-time algorithm.
\end{restatable}

As a consequence, the {\trangle} {\twoclub} is FPT when the considered parameter is vertex cover number (${\vc}$), feedback vertex number (${\fvs}$) etc.
Our next considered parameter is the feedback edge number (${\fes}$) and we prove the following result.

\begin{restatable}{theorem}{ThmFESTriangle}
\label{thm:fes-kernel-vertex-triangle-result}
For every $r \geq 1$, and for every $s$, {\trangle} {\sclub} parameterized by ${\fes}$ admits a kernel with $3{\fes}$ vertices and $4{\fes} - 1$ edges.
\end{restatable}

As the max-leaf number is upper bounded by $\OO({\fes}^2)$ \cite{Eppstein2015Metric}, our above result implies a linear sized kernel when parameterized by the max-leaf number.
It can be observed that {\trangle} {\twoclub} is FPT when the maximum degree (${\mdeg}$) of the input graph is considered as parameter.
Additionally, it holds that ${\mdeg} \leq 2{\bnw}$.
In particular, both bandwidth (${\bnw}$) and the maximum degree (${\mdeg}$) of a graph are compositional.
Hence, {\trangle} {\twoclub} is FPT but does not admit any polynomial kernel unless {\nka} when parameterized by the bandwidth (respectively, maximum degree) of the input graph.
So, our next considered parameter is the $h$-index (${\hind}$) that is provably smaller than the maximum degree (${\mdeg}$) of the input graph.

\begin{restatable}{theorem}{ThmHIndex}
\label{thm:h-index-XP-algorithm}
For every $r \geq 1$, {\trangle} {\twoclub}  
	can be solved in $2^{\OO(k^4)}n^{2^{\OO(k)}}$-time when $k$ is the $h$-index of the input graph.
\end{restatable}

\section{Preliminaries}
\label{sec:prelims}

\subparagraph*{Sets, Numbers, and Graph Theory.} 
We use ${\bN}$ to denote the set of all natural numbers. 
Given $r \in {\bN}$, $[r]$ denotes the set $\{1,2,\ldots,r\}$.
Throughout this paper, we consider undirected graphs and use the standard notations from Diestel's book \cite{DiestelBook}.
Given a graph $G = (V, E)$, we denote $n = |V|$ and $m = |E|$.
For a subset $S \subseteq V(G)$, let $G[S]$ denote the subgraph induced by $S$.
Given $x, y \in V(G)$, the {\em distance} between $x$ and $y$ in $G$ is defined as the number of edges in a shortest path between $x$ and $y$ in $G$.
The {\em diameter} of $G$ is the maximum distance taken over all pairs of vertices in $G$.
A set of edges $F \subseteq E(G)$ is a {\em feedback edge set} of $G$ if $G - F$ is a forest.
Two vertices $u, v \in V(G)$ are {\em twins} if $N_G(u) \setminus \{v\} = N_G(v) \setminus \{u\}$.
A graph has {\em $h$-index} $k$ if $k$ is the largest number such that the graph has at least $k$ vertices each having degree at least $k$.
For a graph $G$, we use ${\hind}(G)$ to denote the $h$-index of $G$.
A graph $G = (V, E)$ is {\em bipartite} if $V(G)$ can be partitioned into two parts $A$ and $B$ for every edge $uv \in E(G)$, $u \in A$ if and only if $v \in B$.
%

\begin{definition}[Tree Decomposition and Treewidth]
\label{defn:treewidth}
Given an undirected graph $G = (V, E)$, a {\em tree decomposition} of $G$ is a pair $\cT = (T, \XX_{t \in V(T)})$ where $T$ is a tree whose every node $t$ is assigned a vertex subset $X_t \subseteq V(G)$, called a {\em bag} such that the followings are satisfied:
\begin{enumerate}[(i)]
	\item 
	$\bigcup_{t \in V(T)} X_t = V(G)$,
	\item 
	for every edge $uv \in E(G)$, there exists $t \in V(T)$ such that $u, v \in X_t$, and 
	\item 
	for every $u \in V(G)$, $T_u = \{t \in V(T) \mid u \in X_t\}$ 
	 induces a connected subgraph of $T$.
\end{enumerate}
The {\em width} of a tree decomposition $\cT = (T, \XX_{t \in V(T)})$ is $\max_{t \in V(T)} |X_t| - 1$.
The {\em treewidth} of $G$, denoted by ${\trw}(G)$ is the minimum possible width of a tree decomposition of $G$.
\end{definition}

It has been observed that a tree decomposition with special properties is more helpful for dynamic programming. 
Given a graph $G = (V, E)$, a tree decomposition $\cT = (T, \XX_{t \in V(T)})$ is {\em nice} if $T$ is a binary tree rooted at a node $r \in V(T)$ and the following conditions are satisfied:
\begin{itemize}
	\item $X_r = \emptyset$ and if $t \in V(T)$ is a leaf node, then $X_t = \emptyset$. 
	\item Every non-leaf node $t \in V(T)$ can be of the following three types:
	\begin{itemize}
		\item {\bf Introduce Node:} $t$ has a unique child $t'$ and there is $u \in X_t \setminus X_{t'}$ such that $X_t = X_{t'} \cup \{u\}$; we say that $u$ is {\em introduced} in $t$.
		\item {\bf Forget Node:} $t$ has a unique child $t'$ and there is $u \in X_{t'} \setminus X_t$ such that $X_t = X_{t'} \setminus \{u\}$; we say that $u$ is {\em forgotten} in $t$.
		\item {\bf Join Node:} $t$ has two children $t_1$ and $t_2$ such that $X_t = X_{t_1} = X_{t_2}$.
	\end{itemize} 
\end{itemize}



\subparagraph*{Parameterized Complexity and Kernelization:} 
A {\em parameterized problem} $L$ is a subset of $\Sigma^* \times \bN$ where $\Sigma$ is a finite alphabet.
The input instance to a parameterized problem $L$ is denoted by $(x, k)$ where $x \in \Sigma^*$ and $k \in \bN$ is the {\em parameter}.
A parameterized problem $L$ is {\em fixed-parameter tractable} if there is an algorithm $\cA$ and a computable function $f: \bN \rightarrow \bN$ that given $(x, k) \in \Sigma^* \times \bN$, runs in $f(k)\cdot |x|^{\OO(1)}$-time and correctly decides whether $(x, k) \in L$.
The algorithm $\cA$ in the above definition is known as {\em fixed-parameter algorithm} (or {\em FPT algorithm}).
Design of preprocessing rules is an important steps to design efficient parameterized algorithms.
A parameterized problem $L \subseteq \Sigma^* \times \bN$ admits a {\em kernelization} if there exists a polynomial-time algorithm $\cB$ that given an input instance $(x, k)$ of $L$, outputs an equivalent instance $(x',k')$ such that $|x'| + k' \leq g(k)$ for some computable function $g: \bN \rightarrow \bN$ and the output instance $(x',k')$ is known as the {\em kernel} for $L$. 
If $g(\cdot)$ is a polynomial function, then $L$ admits a {\em polynomial kernel}.

We use the following conjecture.

\begin{conjecture}[Exponential Time Hypothesis (ETH)]
\label{thm:eth}(\cite{CyganFKLMPPS15})
{\sc $3$-CNF-SAT} cannot be solved in $\OO^*(2^{o(n)})$ time where the input formula has $n$ variables and $m$ clauses.
\end{conjecture}

\subparagraph*{Some Simple Results.} 
Before proving our main results, we first provide some intriguing results on the algorithmic complexity of {\trangle} {\sclub}. 
Hartung et al. (Theorem 3 from \cite{HartungKN15}) proved that {\twoclub} is NP-complete if the input graph has one vertex deletion of which results in a bipartite graph.
We complement this result by proving the following result for {\rVT} {\sclub}.

\begin{proposition}
\label{lemma:bipartite-algorithm}
For every $r \geq 1$, {\rVT} {\sclub} can be solved in polynomial-time if the input graph has one vertex whose deletion results in a bipartite graph.
\end{proposition}

\begin{proof}
Let $G = (V, E)$ be an undirected graph and $x \in V(G)$ such that $G - \{x\}$ is a bipartite graph.
Then, there is a vertex bipartition $(A, B)$ of $G - \{x\}$ such that $A \cap B = \emptyset$, and for every edge $uv \in E(G - \{x\})$, $u \in A$ if and only if $v \in B$.
What is crucial here is that any solution $S$ must contain a triangle.
Since every triangle of $G$ must contain $x$, it must be that $x$ is part of the potential solution $S$.
Hence, for every $u \in A \cup B$, if $u$ is part of the potential solution $S$, then $u$ must be adjacent to $v \in A \cup B$ and $x$ must be adjacent to both $u$ and $v$.
Let $D$ denote the set of vertices $u$ from $N_G(x)$ such that $u$ has a neighbor $v$ in $N_G(x)$.
We consider the set $D \cup \{x\}$.
Observe that between $x$ is adjacent to every vertex of $D$, and between every $u, v \in D$, $x$ is a common neighbor.
Hence, the distance between every pair of vertices is at most 2.
Now, we check the graph $G[D \cup \{x\}]$ and identify which vertex is part of less than $r$ triangles in $G[D \cup \{x\}]$.
If there is any such vertex, then we delete that vertex.
We continue this procedure until we find that every vertex of $D$ is part of at least $r$ triangles in $G[D \cup \{x\}]$.
If $|D \cup \{x\}]| \geq \ell$, then we output a yes-instance, and otherwise we output a no-instance.
\end{proof}


Our above result says that not every hardness result for {\twoclub} transfers to a hardness result of {\trangle} {\twoclub}.
It follows from Theorem \ref{thm:triangle-tw-dp} that {\vtrangle} {\twoclub} is FPT when parameterized by the vertex cover number (${\vc}$).
For the sake of completeness, we provide a proof of nonexistence of polynomial kernels for {\vonetriangle} {\twoclub} parameterized by ${\vc}$ unless {\nka}.
We use the same construction as in \cite{HartungKN15} to prove our result.

\begin{proposition}
\label{prop:vc-lower-bound}
{\vonetriangle} {\twoclub} parameterized by ${\vc}$ does not admit a polynomial kernel unless {\nka}.
\end{proposition}

\begin{proof}
We provide a polynomial parameter transformation from {\clique} parameterized by the vertex cover number (${\vc}$) of the input graph.
Consider $(G, X, \ell)$ the input instance of {\clique} such that $X$ is a vertex cover of $G$.
The question asked is whether $G$ has a clique with at least $\ell$ vertices or not.
Without loss of generality, we assume that $\ell > 3$, as otherwise {\clique} is polynomial-time solvable.

\noindent
{\bf Construction:} Our construction works the same as Hartung et al. \cite{HartungKN15}.
Given a graph $G$ and a vertex cover $X$, let $X = \{v_1,v_2,\ldots,v_k\}$.
We construct a graph $G' = (V', E')$ as follows.
For each $v_i \in X$, create a set of $n$ vertices $Y_i = \{v_i^1,\ldots,v_i^n\}$.
Let $Y = \cup_{i=1}^k Y_i$.
Next, for every edge $e = v_i v_j \in E(G[X])$, we create a vertex $e_{i,j}$ and add to a set $V_E$.
Make $e_{i,j}$ adjacent to every vertex of $Y_i \cup Y_j$ and make $V_E$ into a clique.
In simple terms, the edges between $Y$ and $V_E$ are the vertex-edge incidence graph of $G[X]$.
Subsequently, add the vertices of $Z = V(G) \setminus X$ into $V'$.
For every $u \in Z$, $u$ is adjacent to $e_{i, j}$ if and only if $uv_i, uv_j \in E(G)$.
This completes the construction of $G'$ and we set $\ell' = (\ell - 1)n + |V_E| + 1$.

Observe that $V_E$ is a vertex cover of $G'$ and $|V_E| \leq {{|X|}\choose{2}}$.
We claim that $(G, X, \ell)$ is a yes-instance if and only if $(G', V_E, \ell')$ is a yes-instance for {\vonetriangle} {\twoclub}.

($\Rightarrow$) For the forward direction, let $S$ be a clique with at least $\ell$ vertices.
There are two cases.
If $S \subseteq X$, then we set $T = (\cup_{v_i \in S} Y_i) \cup V_E$.
Clearly, $|T| = |V_E| + |S|n = |V_E| + n\ell > |V_E| + (\ell - 1)n + 1$.
We can use arguments similar to \cite{HartungKN15} to prove that $T$ is a 2-club.
It remains to justify that $T$ is a 2-club of $G'$ satisfying vertex 1-triangle property.
We exploit the fact that $\ell > 4$ and $r = 1$.
Consider $e_{i,j} \in V_E$.
Since $|S| \geq 4$, it must be that $S \cap X$ has at least 4 vertices.
Then, $V_E$ must have at least 3 vertices, hence $e_{i,j}$ is part of at least one triangle.
Now, consider $v_i^j \in Y$.
Note that $v_i \in S$.
Now, $v_i$ is adjacent to 3 vertices in $S$ and hence $v_i$ is incident to 3 edges in $G[S]$.
It means that $v_i^j$ is adjacent to 3 vertices in $V_E$.
Since $V_E$ is a clique, it follows that $v_i^j$ is part of at least one triangle.
Hence, $T$ satisfies vertex 1-triangle property.

The other case occurs when $S \setminus X$ has a vertex $u$.
Observe that at most one such vertex  can exist from $S$.
Then, we set $T = (\cup_{v_i \in S} Y_i) \cup V_E \cup \{u\}$.
Since $|S \cap X| = \ell -1$, note that $|T| = (\ell - 1)n + 1 + |V_E|$.
It follows due to the arugments by Hartung et al. \cite{HartungKN15} that $T$ is a 2-club of $G'$.
We justify that $T$ satisfies vertex 1-triangle property.
As $|S \cap X| > 2$, it must be that $|V_E| \geq 3$.
Hence, for every $e_{i,j} \in V_E$ and for every $v_i^j \in Y_i$, we can prove using similar arguments as before that these vertices are part of at least 1 triangle in $G'[T]$.
Consider $u \in T$.
Note that $u$ is adjacent to 3 vertices in $X$.
Then, there are three edges $v_i v_j \in G[X]$ such that $u$ is adjacent to $e_{i,j} \in T$.
Therefore, $u$ is part of at least one triangle in $G'[T]$.
It implies that $u$ is part of at least one triangle in $G'[T]$.
This completes the proof of forward direction.

For the backward direction ($\Leftarrow$), let $T$ be a 2-club of $G'$ that satisfies vertex 1-triangle property and $|T| \geq (\ell-1)n + |V_E| + 1$.

Let $T_Y = T \cap Y$. We have $|T_Y| \geq |T \setminus (V_E \cup Z) | \geq (\ell-1)n + |V_E| + 1 - (|V_E| + n-k) \geq (\ell-2)n+1$. As each group $Y_i$ is of size $n$, by pigeon-hole principle, there exist at least $\ell-1$ groups $Y_i$ such that $T_Y$ contains at least one vertex from the group. Let $I$ denote the corresponding set of indices.

For each $i, j \in I$, let $a \in Y_i \cap T$ and $b \in Y_j \cap T$. Since $T$ is a $2$-club, there is a path of length at most $2$ between $a$ and $b$ in $G'[T]$. This can only happen if there are edges from $e_{i,j} \in V_E$ to $a$ and $b$. Thus all vertices $v_i, v_j$ with $i, j \in I$ are adjacent in $G$. Thus the set $C = \{v_i: i \in I\}$ forms a clique of size at least $\ell-1$.
Note that if $T \cap Z = \emptyset$, then  $|T_Y| \geq |T \setminus V_E  | \geq (\ell-1)n + 1$, implying that there are $\ell$ such groups. In this case, $C$ forms a clique of size at least $\ell$. Else, there exist a vertex $z \in T \cap Z$. As $z$ is part of the $2$-club, for each $y \in T \cap Y_i$, there is a path of length at most $2$ between $z$ and $y$ in $G'[T]$. This can only happen if there are edges from $e_{z,i} \in V_E$ to $z$ and $a$. Thus $z$ is adjacent to all the vertices of $C$ in $G$, forming a clique of size at least $\ell$.

Observe that this is a polynomial parameter transformation from {\clique} parameterized by the vertex cover number to {\vonetriangle} {\twoclub} parameterized by the vertex cover number (quadratic blow up).
Hence, {\vonetriangle} {\twoclub} parameterized by the vertex cover number (${\vc}$) admits no polynomial kernel unless {\nka}.
\end{proof}


\section{Parameterized by the treewidth}
\label{sec:FPT-treewidth}


In this section, we provide our results when parameterized by the treewidth of the input graph.
First, we describe how to design an FPT algorithm for {\trangle} {\twoclub} when parameterized by ${\trw}$, that is, the treewidth (${\trw}$) of the graph.
We first design an algorithm for the combined parameter ${\trw} + r$.
Subsequently, we explain how that algorithm will imply an FPT algorithm when parameterized by the treewidth ($\trw$) alone.


\subparagraph*{Dynamic Programming:} Let $\cT = (T, \{X_t\}_{t \in V(T)})$ be a nice tree decomposition of $G$ of minimum width ${\trw}(G)$ and root $r$.
We are ready to describe the details of the dynamic programming over $\cT$.
Let $T_t$ denote the subtree of $T$ rooted at node $t \in V(T)$ and $G_t$ denote the subgraph induced on the vertices that are present in the bags of $T_t$.
Consider a 2-club $S$ of $G$ that satisfies {\vtrangle} and a {\em partial solution} $S_t$ by restricting $S$ to the subgraph $G_t$.
Formally, $S_t = S \cap V(G_t)$.
As $G_t[S_t]$ is the same graph as $G[S_t]$, we use $G[S_t]$ for simplicity.
Before formally defining the subproblem, we first focus on the properties that we need for the vertex subset $S_t$.
Our next definition formalizes those properties.

\begin{definition}[Valid Tuple and Candidate Solution]
\label{defn:tw-dp-tuple-candidate-soln}
Consider $A \subseteq X_t$ such that $f: A \rightarrow [r] \cup \{0\}$ and $\cA \subseteq 2^A$.
Then, $\tau(A) = (f, \cA)$ is a {\em valid tuple} for $t$ if there exists $S_t \subseteq V(G_t)$ that satisfies each of the following conditions.
\begin{enumerate}[(i)]
	\item\label{prop-bag-intersection} $A = S_t \cap X_t$.
	\item\label{prop-forgotten-dist-2-part-1} For every $x \in S_t \setminus A$, and for every $y \in S_t$, the distance between $x$ and $y$ must be at most 2 in $G[S_t]$.
	\item\label{prop-number-of-triangles-intersection} For every $v \in A$, $f(v)$ is used to represent the number of triangles in $G[S_t]$ that $v$ participates in. This value can be less than $r$, as there may exist some neighbors of $v$ in $V(G) \setminus V(G_t)$ that can increase the number of triangles that $v$ is part of. We allow these values to vary between $0$ and $r$.
	If $f(v) < r$, then $v$ is part of $f(v)$ triangles in $G_t[S_t]$.
	If $f(v) = r$, then $f(v)$ is part of at least $r$ triangles in $G_t[S_t]$.
	\item\label{prop-number-of-triangles-forgotten} For every $v \in S_t \setminus A$, $v$ is part of at least $r$ triangles in $G[S_t]$.
	\item\label{prop-subset-information-forgot-vertex} A set $B \subseteq A$ is in $\cA$ if and only if there exists $v \in S_t \setminus A$ such that $N_{G_t}(v) \cap A = B$.
	In simple terms, $\cA$ stores the information of the neighborhood of every vertex in $S_t \setminus A$.
\end{enumerate}

Given $A \subseteq X_t$, a subset $S_t \subseteq V(G_t)$ that satisfies all conditions described in items (\ref{prop-bag-intersection})-(\ref{prop-subset-information-forgot-vertex}) is called a {\em candidate solution} for the node $t$ with respect to the tuple $\tau(A)$.
If a node $t$ has no candidate solution with respect to a tuple $\tau(A)$, then $\tau(A)$ is {\em invalid}.
Given $A \subseteq X_t$, if $S_t$ is a candidate solution of maximum possible cardinality for a node $t$ with respect to tuple $\tau(A)$, then we say that $R_t$ is an {\em optimal} candidate solution for node $t$ with respect to the tuple $\tau(A)$.
\end{definition}

We define a subproblem ${\dptw}[t, A, \tau(A)]$ that stores the size of a maximum cardinality subset $S_t \subseteq V(G_t)$ that satisfies all the conditions described in items (\ref{prop-bag-intersection})-(\ref{prop-subset-information-forgot-vertex}).
Given $A \subseteq X_t$, if $\tau(A)$ is invalid, then we denote ${\dptw}[t, A, \tau(A)] = -\infty$.

\begin{lemma}
\label{lemma:treewidh-DP-table}
Given a nice tree decomposition $\cT = (T, \{X_t\}_{t \in V(T)})$ of $G$ of width $k$, there exists an algorithm that correctly compute the values of all subproblems in $r^{\OO({k})}2^{2^{\OO({k})}}n^{\OO(1)}$-time. 
\end{lemma}

\begin{proof}
When a nice tree decomposition $\cT$ is given, a node $t \in V(T)$ is a leaf node, or an introduce node, or a forget node, or a join node.
For the sake of simplicity, given a vertex subset $S \subseteq V_t$, we use $E(S)$ to denote the set of edges in $G_t[S]$
We use bottom-up dynamic programming algorithm to compute the values of these subproblems as follows.
Our bases cases occur at every leaf node.

\begin{description}
	\item[(i)] 
	{\bf Leaf Node}: If $t$ is a leaf node, then $X_t = \emptyset$. 
	Then, the only possible subset $A = \emptyset$ and the only possible function $f: A \rightarrow [r] \cup \{0\}$ is the empty function.
	Hence, a tuple $\tau(A)$ is either $(f, \emptyset)$ or $(f, \{\emptyset\})$ with the condition that $f$ is an empty function.
	Observe that $(f, \{\emptyset\})$ is invalid for $t$ because $\emptyset$ is the only vertex subset of $G_t$ and there is no vertex outside the bag such that $\emptyset \in \cA$.
	Hence, we have the following two situations.
	\begin{equation}
	\label{eq:leaf-node}
	{\dptw}[t, \emptyset, \tau(\emptyset)] = 
	\begin{cases}
		0  & \text{ if } \tau(\emptyset) = (\emptyset, \emptyset) \\
		-\infty  & \text{ otherwise }
	\end{cases} 
	\end{equation}  
	
	\item[(ii)] 
	{\bf Introduce Node}: Let $X_t$ be an introduce node and $X_{t'}$ be its unique child such that $v \in X_t \setminus X_{t'}$.
	We need the following definition to compute ${\dptw}[t, A, \tau(A)]$ using the values of the subproblems in its child.
	
	\smallskip
	
	\begin{definition}[Introduce-$v$-compatible]
	\label{defn:introduce-A-prime-compatible}
	Let $t \in V(T)$ be an introduce node with its unique child $t' \in V(T)$ such that  $X_t \setminus X_{t'} = \{v\}$, and $A \subseteq X_t$.
	Suppose that $v \in A$ and $\tau(A) = (f, \cA)$ is a valid tuple for $t$, $A' = A \setminus \{v\}$, and $\gamma(A') = (f', \cA')$ is a valid tuple for $t'$.
	Furthermore, for $x, y \in A$, let $r_A(t, x, y)$ denote the number of triangles in $G_t[A]$ that contain both $x$ and $y$. 
	Then, $\tau(A)$ is {\em introduce-$v$-compatible} with $\gamma(A')$ if 
	\begin{enumerate}
		\item\label{add-v-cond-1} for every $u \in A'$, $f(u) = \min\{r, f'(u) + r_A(t, u, v)\}$,
		\item\label{add-v-cond-2} $f(v) = \min\{r, |E(N_{G_t}(v) \cap A)|\}$,
		\item\label{add-v-cond-3} $\cA = \cA'$, and
		\item\label{add-v-cond-4} for all $B \in \cA'$, $B \cap N_{G_t}(v) \neq \emptyset$.
	\end{enumerate}
	\end{definition}
	
	 
	Consider an introduce node $t$ with unique child $t'$ such that $X_t \setminus X_{t'} = \{v\}$.
	There are two cases that can arise whether $v \in A$ or not.
	If $v \in A$, then we use $\cB_v(\tau(A))$ to denote the collection of all valid tuples $\gamma(A') = (f', \cA')$ for $t'$ that are introduce-$v$-compatible with $\tau(A)$.
	We are ready to state our next claim that provides a recurrence for introduce node.
	
	
	\begin{claim}
	\label{claim:intro-node-correctness}
	Let $t$ be an introduce node, $A \subseteq X_t$ with child $t'$ such that $X_t \setminus X_{t'} = \{v\}$, and $\tau(A)$ be a tuple for $t$.
	\begin{itemize}
		\item If $v \in A$, then let $A' = A \setminus \{v\}$.
		If ${\cB}_v(\tau(A)) \neq \emptyset$, then 
		$${\dptw}[t, A, \tau(A)] =
		1 + \max\limits_{\gamma(A') \in {\cB}_v(\tau(A))} {\dptw}[t', A', \gamma(A')]$$
		If ${\cB}_v(\tau(A)) = \emptyset$, then ${\dptw}[t, A, \tau(A)] = -\infty$.
		
		\item If $v \notin A$, then 
		${\dptw}[t, A, \tau(A)] = {\dptw}[t', A, \tau(A)]$.
	\end{itemize} 
	Finally, ${\dptw}[t, A, \tau(A)]$ can be computed in $r^{\OO(k)}n^{\OO(1)}$-time.
	\end{claim}

	\begin{proof}
	We only establish the correctness of the recurrence relations.
In order to establish that, we need to prove two auxiliary claims.
The first one is the following.

\begin{nclaim}
	\label{claim:intoduce-v-validity}
	Let $t \in V(T)$ be an introduce node and its unique child $t'$ such that $X_t \setminus X_{t'} = \{v\}$, $A \subseteq X_t$ such that $v \in A$, $\tau(A)$ is a tuple for $t$, and $A' = A \setminus \{v\}$.
	If $\tau(A)$ is a valid tuple for $t$ and $R$ is a candidate solution for $t$ with respect to $\tau(A)$, then there exists a valid tuple $\gamma(A') = (f', \cA')$ for $t'$ such that 
	\begin{itemize}
		\item $\tau(A)$ is introduce-$v$-compatible with $\gamma(A')$, and
		\item $R \setminus \{v\}$ is a candidate solution for $t'$ with respect to $\gamma(A')$.
	\end{itemize} 
	\end{nclaim}
	
	\begin{proof}
	Consider an introduce node $t$ and its unique child $t'$ such that $X_t \setminus X_{t'} = \{v\}$, $A \subseteq X_t$ such that $v \in A$, and $\tau(A)$ is a tuple for $t$.
	Assume that $A' = A \setminus \{v\}$.
	
	Let $\tau(A) = (f, \cA)$ be a valid tuple for $t$ and $R \subseteq V(G_t)$ is a candidate solution for $t$ with respect to $\tau(A)$.
	Since $G_t[R] = G[R]$, we use $G[R]$ for simplicity.
	We provide the existence of a tuple $\gamma(A') = (f', \cA')$ for $t'$ as follows.
	First, we set $\cA' = \cA$.
	It ensures that item-(\ref{add-v-cond-3}) of Definition \ref{defn:introduce-A-prime-compatible} is satisfied.
	Then, for every $u \in A'$, we compute $r_A(t, u, v)$, the number of triangles in $G_t[A]$ that both $u$ and $v$ are part of.
	For every $u \in A'$, if $f(u) < r$, then we set $f'(u) = f(u) - r_A(t, u, v)$.
	It ensures that item-(\ref{add-v-cond-1}) of Definition \ref{defn:introduce-A-prime-compatible} is satisfied for every $u \in A'$ such that $f(u) < r$.
	Since $R$ is a candidate solution for $t$ with respect to $\tau(A)$, clearly $f(v)$ must be the number of triangles in $G[R]$ that $v$ is part of subject to a maximum of $r$.
	This number is clearly the same as the number of edges in the neighborhood of $v$ in $A$ subject to a maximum of $r$.
	Therefore, item-(\ref{add-v-cond-2}) is satisfied.
	
	\medskip
	
	We focus on $R' = R \setminus \{v\}$.
	As $R$ is candidate solution for $t$ with respect to $\tau(A)$, due to item-(\ref{prop-number-of-triangles-forgotten}), every $x \in R \setminus A$ is part of at least $r$ triangles in $G[R]$.
	Note that any triangle of $G[R]$ that $x \in R \setminus A$ is part of cannot contain $v$ as $x$ is not adjacent to $v$.
	Therefore, every such triangle of $G[R]$ avoids $v$, implying that such triangles are also in $G[R']$.
	As $R \setminus A = R' \setminus A'$, and every such triangle that $x \in R' \setminus A'$ is part of avoids $v$, hence item-(\ref{prop-number-of-triangles-forgotten}) is satisfied for $R'$.
	
	Due to item-(\ref{prop-forgotten-dist-2-part-1}), every $x \in R \setminus A$ is at distance at most 2 from every $y \in R$ in $G[R]$.
	Consider any arbitrary $x \in R'$ and $y \in R' \setminus A'$.
	As $v$ is not adjacent to $y \in R' \setminus A'$, observe that a path of length 2 between $x$ and $y$ in $G[R]$ does not pass through $v$.
	Hence, a path of length 2 between $x$ and $y$ in $G[R]$ is also a path of length 2 between $x$ and $y$ in $G[R']$.
	It implies that item-(\ref{prop-forgotten-dist-2-part-1}) of Definition \ref{defn:tw-dp-tuple-candidate-soln} is satisfied for $R'$.
	
	Due to item-(\ref{prop-forgotten-dist-2-part-1}), $v$ is of distance at most 2 to every $u \in R \setminus A$ in $G[R]$.
	Due to item-(\ref{prop-subset-information-forgot-vertex}), for every $u \in R \setminus A$, $N_{G_t}(u) \cap A \in \cA$.
	As $N_{G_t}(v) \subseteq X_t$, it must be that $v$ is of distance exactly 2 from every $u \in R \setminus A$.
	It means that $v$ has a common neighbor with every $u \in R \setminus A$ in $G[R]$.
	Therefore, if $B \in \cA$, then there is $u \in R \setminus A$ such that $N_{G_t}(u) \cap A = B$.
	As $u$ is of distance at most 2 from $v$, hence $v$ must have a neighbor $w$ in $B$.
	Hence, for every $B \in \cA$, $N_{G_t}(v) \cap B \neq \emptyset$.
	It satisfies item-(\ref{add-v-cond-4}) of Definition \ref{defn:introduce-A-prime-compatible}.
	
	Note that if $u \in R' \setminus A'$, then $N_{G_t}(u) \cap A' = N_{G_t} \cap A \in \cA$.
	Hence, $N_{G_t}(u) \cap A' \in \cA'$, satisfying item-(\ref{prop-subset-information-forgot-vertex}) for $R'$.
	
	Now, for every $u \in A \setminus \{v\}$, we consider the case when $f(u) = r$.
	Notice that $r_A(t, u, v)$ is the number of triangles that both $u$ and $v$ are part of.
	If $f(u) = r$, then it is possible that $u$ can be part of at least $r$ triangles in $G[R]$.
	Out of these triangles, we focus on the number of triangles that are in $G[R']$.
	If $u$ is part of at least $r$ triangles in $G[R']$, then $f'(u) = r$.
	Otherwise, $f'(u)$ is set to the exact number of triangles ($\leq r-1$) in $G[R']$ that $u$ is part of.
	It ensures that item-(\ref{add-v-cond-1}) of Definition \ref{defn:introduce-A-prime-compatible} is satisfied for every $u \in A'$ such that $f(u) = r$.
	As $f'(u)$ can be assigned some value from $\{0,1,\ldots,r\}$ for every $u \in A'$, it also ensures us that item-(\ref{prop-number-of-triangles-intersection}) is satisfied for $R'$.
	Finally, $R' \cap V(G_{t'}) = A'$. It implies that item-(\ref{prop-bag-intersection}) is satisfied for $R'$.
	
	As we have provided the existence of a tuple $\gamma(A') = (f', \cA')$ satisfying items-(\ref{add-v-cond-1})-(\ref{add-v-cond-4}) of Definition \ref{defn:tw-dp-tuple-candidate-soln}, and $R'$ is a subset satisfying the properties of item-(\ref{prop-bag-intersection})-(\ref{prop-subset-information-forgot-vertex}) of Definition \ref{defn:tw-dp-tuple-candidate-soln}, this completes the proof.
	\end{proof}
	
	Our next claim provides a converse variant of the previous claim.
	 
	\begin{nclaim}
	\label{claim:introduce-validity-child} 
	Let $t \in V(T)$ be an introduce node and its unique child $t'$ satisfying $X_t \setminus X_{t'} = \{v\}$, and $A \subseteq X_t$ such that $v \in A$ and $A' = A \setminus \{v\}$.
	Furthermore, $\gamma(A')$ is a valid tuple for $t'$ and $R'$ is a candidate solution for $t'$ with respect to $\gamma(A')$.
	If $\tau(A) = (f, \cA)$ is a tuple for $t$ such that $\tau(A)$ is introduce-$v$-compatible with $\gamma(A')$, then $R' \cup \{v\}$ is a candidate solution for $t$ with respect to $\tau(A)$.
	\end{nclaim}
	
	\begin{proof}
	Consider the introduce node $t$ and its unique child $t'$ such that $X_t \setminus X_{t'} = \{v\}$, $A \subseteq X_t$ such that $v \in A$. and $A' = A \setminus \{v\}$.
	Let $\gamma(A') = (f', \cA')$ be a valid tuple for $t'$ and $R'$ be a candidate solution for $t'$ with respect to $\gamma(A')$ and $\tau(A) = (f, \cA)$ such that $\tau(A)$ is introduce-$v$-compatible with $\gamma(A')$.
	We consider $R = R' \cup \{v\}$.
	
	Based on the assumption in the statement, $\tau(A)$ is introduce-$v$-compatible with $\gamma(A')$.
	Hence, due to item-(\ref{add-v-cond-4}), every $B \in \cA'$ contains a vertex that is a neighbor of $v$.
	As $\cA'$ contains the neighborhood of every $x \in R' \setminus A'$, it follow that every $x \in R' \setminus A'$, there is a common neighbor between $v$ and $x$ in $A'$.
	As $R \setminus A = R' \setminus A'$, for every $x \in R \setminus A$, the distance between $x$ and $v$ is at most 2 in $G[R]$.
	Also, $v$ is not adjacent to any $x \in R \setminus A$.
	As $R'$ is a candidate solution for $t'$ with respect to $\gamma(A')$, due to item-(\ref{prop-forgotten-dist-2-part-1}) of Definition \ref{defn:tw-dp-tuple-candidate-soln}, every $x \in R' \setminus A'$ has distance at most 2 to every $w \in R'$ in $G[R']$.
	As $G[R']$ is a subgraph of $G[R]$, the same path exists in $G[R]$ as well.
	Hence every $x \in R \setminus A$ has distance at most 2 to every $w \in R$, satisfying item-(\ref{prop-forgotten-dist-2-part-1}) of Definition \ref{defn:tw-dp-tuple-candidate-soln} for $R$.
	
	Clearly, $R \cap V(G_t) = A$, hence item-(\ref{prop-bag-intersection}) is satisfied for Definition \ref{defn:tw-dp-tuple-candidate-soln}.
	Next, we check every $u \in A$ one by one.
	As $\tau(A)$ is introduce compatible with $\gamma(A')$, this function $f: A \rightarrow [r]\cup \{0\}$ satisfies item-(\ref{add-v-cond-2}) of Definitio \ref{defn:tw-dp-tuple-candidate-soln}.
	Hence, $f(v)$ is the number of edges that are among the neighbors of $v$ intersecting $A$.
	As $N_{G_t}(v) \subseteq A$, if $v$ is part of a triangle $\{v, x, y\}$ in $G_t[R]$, then $xy$ must be an edge such that $x, y \in N_{G_t}(v) \subseteq A$.
	Since by the definition of $\tau(f, \cA)$, $f(v)$ is the number of edges in $E(N_{G_t}(v) \subseteq A)$ subject to a maximum of $r$, note that $f(v)$ correctly sets the number of triangles that contains $v$. 
	For every other $u \in A \setminus \{v\}$, based on item-(\ref{add-v-cond-1}) of Definition \ref{defn:introduce-A-prime-compatible}, $f(u) = \min\{r, f'(u) + r_A(t, u, v)\}$.
	Now, let us examine the number of triangles in $G_t[R]$ that contains $u$.
	It is the sum of the number of triangles in $G_t[R']$ containing $u$ and the number of triangles in $G_t[R]$ that contains both $u$ and $v$.
	The number of triangles in $G_t[R]$ containing both $u$ and $v$ is $r_A(t, u, v)$.
	As $R'$ is a candidate solution for $t'$ with respect to $(f', \cA')$,  $u$ participates in $f'(u)$ triangles in $G_t[R']$.
	As $f(u)$ is assigned $f'(u) + r_A(t, u, v)$ subject to a maximum of $r$, $f(u)$ correctly sets the number of triangles in $G_t[R]$ that $u$ is part of.
	Hence, item-(\ref{prop-number-of-triangles-intersection}) is satisfied for $R$.
	
	As $\tau(A)$ is introduce compatible with $\gamma(A')$, due to item-(\ref{add-v-cond-3}), it follows that $\cA = \cA'$.
	Since $\cA$ stores the neighborhood for every $x \in R \setminus A$ in $R'$, due to item-(\ref{add-v-cond-4}) for every $x \in R \setminus A$, $x$ has a neighbor in $R$ that is also a neighbor of $v$.
	Hence, every $x \in R \setminus A$ has distance 2 with $v$.
	As for every $x \in R \setminus A$, $B_x = N_{G_t}(x) \cap A' \in \cA'$, $\cA' = \cA$, hence $N_{G_t}(v) \cap B_x \neq \emptyset$.
	Therefore, item-(\ref{prop-subset-information-forgot-vertex}) from Definition \ref{defn:tw-dp-tuple-candidate-soln} are satisfied.
	
	Finally, we consider the number of triangles in $G_t[R]$ that every $x \in R \setminus A$ is part of.
	Since every $x \in R \setminus A$ is also in $R' \setminus A'$, and $R'$ is a candidate solution for $t'$ with respect to $\gamma(A')$, $x$ is part of at least $r$ triangles in $G_{t'}[R']$.
	And none of these $r$ triangles contain $v$.
	Therefore, every $x \in R \setminus A$ is part of at least $r$ triangles even in $G_t[R]$ as that is a supergraph of $G_{t'}[R']$.
	Hence, $R'$ satisfies item-(\ref{prop-number-of-triangles-forgotten}) from Definition \ref{defn:tw-dp-tuple-candidate-soln}.
	
	As $\tau(A)$ satisfies all the conditions with $\gamma(A')$ from Definition \ref{defn:introduce-A-prime-compatible}, this proves the second item of this claim.
	As $R$ satisfies the conditions from Definition \ref{defn:tw-dp-tuple-candidate-soln}, this proves the first item of this claim.
	\end{proof}

	Now, using the above two claims, we give the proof for two cases, $v \in A$ and $v \notin A$.
	
	\begin{itemize}
		\item Case $v \in A$. 
		%
		Consider a solution $R$ for $t$ with respect to $\tau(A)$.
		Then, due to Claim \ref{claim:intoduce-v-validity}, there exists a valid tuple $\gamma(A')$ for $t'$ such that $R' = R \setminus \{v\}$ is a candidate solution for $t'$ with respect to $\gamma(A')$.
		As $|R'| \leq {\dptw}[t', A', \gamma(A')]$, and ${\dptw}[t', A', \gamma(A')] \leq \max_{\gamma(A') \in {\cB}_v(\tau(A))} {\dptw}[t', A', \gamma(A')]$, the following inequality holds true. 
	
		$${\dptw}[t, A, \tau(A)] = |R| = 1 + |R'| \leq 1 + \max_{\gamma(A') \in {\cB}_v(\tau(A))} {\dptw}[t', A', \gamma(A')]$$

		Conversely, let $R' \subseteq V(G_{t'})$ such that $|R'| = \max_{\gamma(A') \in {\cB}_v(\tau(A))} {\dptw}[t', A', \gamma(A')]$.
		Let $\beta(A') = (f', \cA')$ be the tuple for $t'$ such that $\tau(A)$ is introduce compatible with $\beta(A')$ and $R'$ be an optimal candidate solution for $t'$ with respect to $\beta(A')$.
		Then, it follows from Claim \ref{claim:introduce-validity-child}, $R' \cup \{v\}$ is a candidate solution for $t$ with respect to $\tau(A)$.
		Hence, $|R'| + 1 \leq {\dptw}[t, A, \tau(A)]$, and we have
		
		$$1 + \max_{\gamma(A') \in {\cB}_v(\tau(A))} {\dptw}[t', A', \gamma(A')] = 1 + |R'| = |R| \leq {\dptw}[t, A, \tau(A)]$$ 
		
		This completes our proof when $v \in A$.
		
		\medskip
		
		\item Case $v \notin A$. Then, $A' = A$.
		If $R$ is a candidate solution for $t$ with respect to $\tau(A) = (f, \cA)$, observe that $G_t[R]$ and $G_{t'}[R]$ are the same graphs.
		Hence, $R$ is a candidate solution for $t'$ with respect to $(f, \cA)$.
		Hence, $|R| \leq {\dptw}[t', A, \tau(A)]$.
		By using similar arguments, we can prove that $|R| \geq {\dptw}[t', A, \tau(A)]$.
	\end{itemize}
	As the cases are mutually exhaustive, this completes the proof of correctness for both the cases.	
	\end{proof}
	
	\smallskip
	
	\item[(iii)] 
	{\bf Forget node}: Let $X_t$ be a forge node and $t'$ be its unique child such that $X_{t'} = X_t \cup \{v\}$ such that $v \notin X_t$.
	Similar to Definition \ref{defn:introduce-A-prime-compatible}, we define the following term in the context of forget nodes.
	
	
	\begin{definition}[Forget-$v$-Compatible]
	\label{defn:forget-A'-compatible}
	Let $t$ be a forget node with $A \subseteq X_t$ with unique child $t'$ such that $\{v\} = X_{t'} \setminus X_{t}$, $\tau(A) = (f, \cA)$ be a valid tuple for $t$, and $A' = A \cup \{v\}$.
	Consider a valid tuple $\gamma(A') = (f', \cA')$ for $t'$.
	Then, $\tau(A)$ is {\em forget-$v$-compatible} with $\gamma(A')$ if
	\begin{enumerate}
		\item\label{remove-cond-1} $f'(v) = r$ and for every $u \in A$, $f'(u) = f(u)$,
		\item\label{remove-cond-2} for every $B' \in \cA'$, if $v \in B'$, then $B' \setminus \{v\} \in \cA$,
		\item\label{remove-cond-3} for every $B' \in \cA'$, if $v \notin B'$, then $B' \in \cA$, 
		\item\label{remove-cond-4}
		If $N_{G_t}(v) \cap A \in \cA$, and
		\item\label{remove-cond-5} for every $x \in A$, $\{x, v\} \subseteq B$ for some $B \in \cA'$, or the distance between $x$ and $v$ is at most 2 in $G_{t'}[A']$.
	\end{enumerate} 
	\end{definition}
	
	

	
	
	If $\tau(A)$ is a valid tuple for $t$, then let $\ccC_v(\tau(A))$ denote the set of all valid tuples $\gamma(A \cup \{v\})$ such that 
	 $\tau(A)$ is forget-$v$-compatible with $\gamma(A')$.
	 Our next claim gives us the recurrence for forget nodes.
	
	\begin{claim}
	\label{claim:correctness-forget-node}
	If $t$ is a forget node, $A \subseteq X_t$ with child $t'$ such that $X_{t'} \setminus X_t = \{v\}$, $A' \setminus A = \{v\}$, and $\tau(A)$ is a tuple for $t$, 
	then 
	$${\dptw}[t, A, \tau(A)] = \max \bigg\{ {\dptw}[t', A, \tau(A)],  \max\limits_{\gamma(A') \in \ccC_v(\tau(A))} {\dptw}[t', A', \gamma(A')]\bigg\}$$
	
	
	 The value of ${\dptw}[t, A, \tau(A)]$ can be computed in $2^{2^{\OO(k)}} n^{\OO(1)}$-time.	
	\end{claim}

	\begin{proof}
	We need to prove the following two claims that would help us to define the recurrence for forget nodes.
	
	\begin{nclaim}
	\label{claim:forget-v-parent}
	Let $t$ be a forget node with child $t'$ such that $v \in X_{t'} \setminus X_{t}$, $A \subseteq X_t$ and $A' = A \cup \{v\}$.
	If $\tau(A) = (f, \cA)$ is a valid tuple for $t$ and $R$ is a candidate solution for $t$ with respect to $\tau(A)$ such that $v \in R$, then there exists a valid tuple $\gamma(A')$ for $t'$ such that $R$ is a candidate solution for $t'$ with respect to $\gamma(A')$.
	\end{nclaim}
	
	\begin{proof}
	Let $t$ be a forget node with child $t'$ such that $v \in X_{t'} \setminus X_{t}$, $A \subseteq X_t$ and $A' = A \cup \{v\}$.
	Suppose that $\tau(A) = (f, \cA)$ is a valid tuple for $t$ and $R$ is a candidate solution for $t$ with respect to $\tau(A)$ such that $v \in R$.
	Observe that $G_t[R]$ is the same graph as $G_{t'}[R]$.
	Since $v \in A' \cap R$, and $A = R \cap X_t$, follows that $A' = X_{t'} \cap R$.
	
	We justify the existence of a tuple $\gamma(A') = (f', \cA')$ such that $\tau(A)$ is forget-$v$-compatible with $\gamma(A')$.
	For every $u \in A$, we set $f'(u) = f(u)$ and $f(v) = r$.
	This ensures that item-(\ref{remove-cond-1}) from Definition \ref{defn:forget-A'-compatible} is satisfied.
	We check every $B \in \cA$ one by one.
	Note that there is $u \in R \setminus A$ such that $N_{G_t}(u) \cap A = B$.
	For each vertex $u \in R \setminus A$ one by one such that $N_{G_t}(u) \cap A = B$.
	We check whether $u$ is adjacent to $v$ or not.
	If $u$ is adjacent to $v$, then we add $B \cup \{u\}$ into $\cA'$, satisfying item-(\ref{remove-cond-2}).
	Otherwise, $u$ is not adjacent to $v$, then we add $B$ into $\cA'$, satisfying item-(\ref{remove-cond-3}).
	
	Since $R$ is a candidate solution for $t$ with respect to $\tau(A)$, the distance between $v$ and every vertex of $A$ is at most 2, due to item-(\ref{prop-forgotten-dist-2-part-1}).
	Hence, there is a path between $v$ and every vertex $x \in A$.
	Such a path can pass through $A'$ only or through a vertex $y \in R \setminus A'$.
	If a path of length 2 between $x \in A$ and $v$ passes through $A'$ only, then item-(\ref{remove-cond-5}) is satisfied.
	Otherwise, a path of length 2 between $x \in A$ and $v$ passes through $y \in R \setminus A'$.
	Then, $y$ is adjacent to $v$ and $x$.
	Also, $B_y = N_{G_t}(y) \cap A \in \cA$.
	By our construction of $\cA'$, we have added $B_y \cup \{v\}$ into $\cA'$
	Hence, item-(\ref{remove-cond-5}) is satisfied.
	Since $R$ is a candidate solution for $t$ with respect to $\tau(A)$, and $v \in R \setminus A$, hence $N_{G_t}(v) \cap A \in \cA$.
	This ensures that item-(\ref{remove-cond-4}) is satisfied.
	Since all the items from Definition \ref{defn:forget-A'-compatible} is satisfied, hence $\tau(A)$ is forget-$v$-compatible with $\gamma(A')$.
	
	Now, we argue that $R$ is a candidate solution for $t'$ with respect to $\gamma(A')$.
	Observe that $R \cap X_{t'} = A'$.
	This ensures that item-(\ref{prop-bag-intersection}) from Definition \ref{defn:tw-dp-tuple-candidate-soln} is satisfied.
	We focus on $x \in R \setminus A'$ and $y \in A'$.
	Since $R$ is a candidate solution for $t$ with respect to $\tau(A)$, item-(\ref{prop-forgotten-dist-2-part-1}) is satisfied for every $w \in R \setminus A$ and every $z \in R$.
	Observe that $A' = A \cup \{v\}$.
	As $G_t[R]$ and $G_{t'}[R]$ are the same graphs, $R \setminus A' \subseteq R \setminus A$, every $w \in R \setminus A'$ is of distance at most 2 to every $z \in R$.
	Now, $v \in R$.
	Hence, every $x \in R \setminus A'$ is of distance at most 2 from every $z \in R$.
	Therefore, item-(\ref{prop-forgotten-dist-2-part-1}) is satisfied for $R$ in $t'$.
	
	Since $R$ is a candidate solution for $t$ with respect to $\tau(A)$ and $v \notin A$, hence $v$ must be part of at least $r$ triangles in $G_t[R]$.
	As $v \in A'$ and $G_t[R]$ is the same graph as $G_{t'}[R]$, our choice of $f'(v) = r$ correctly ensures that $v$ is part of at least $r$ triangles in $G_{t'}[R]$.
	Since every $u \in A$ is also in $A'$, and $f'(u) = f(u)$, clearly $u$ is part of $f'(u)$ triangles in $G_{t'}[R]$.
	It ensures us that item-(\ref{prop-number-of-triangles-intersection}) is satisfied for $R$.
	For other vertices $w \in R \setminus A'$, $w \in R \setminus A$, and they are already part of at least $r$ triangles in $G_t[R]$.
	Hence, $w$ is part of at least $r$ triangles in $G_{t'}[R]$.
	Therefore, item-(\ref{prop-number-of-triangles-forgotten}) is satisfied for $R$.
	
	Now, we consider the set $\cA'$.
	We consider every $y \in R \setminus A'$ one by one.
	If $B_y = N_{G_t}(y) \cap A \in \cA$, and $v$ is not adjacent to $y$, then we have already put $B_y$ into $\cA'$.
	Similarly, if $v$ is adjacent to $y$, then we put $B_y \cup \{v\}$ into $\cA'$.
	Hence, $\cA'$ stores the neighborhood of every vertex of $R \setminus A'$, satisfying item-(\ref{prop-subset-information-forgot-vertex}) of Definition \ref{defn:tw-dp-tuple-candidate-soln}.
	
	As $R$ satisfies every condition required for Definition \ref{defn:tw-dp-tuple-candidate-soln}, $R$ is a candidate solution for $t$ with respect to $\gamma(A')$.
	\end{proof}
	
	\begin{nclaim}
	\label{claim:forget-v-child}
	Let $t$ be a forget node with child $t'$ such that $v \in X_{t'} \setminus X_{t}$, $A \subseteq X_t$ and $A' = A \cup \{v\}$.
	Furthermore, assume that $\gamma(A') = (f', \cA')$ is a valid tuple for $t'$ and $R$ is a candidate solution for $t'$ with respect to $\gamma(A')$ such that 
		$v \in R$, $f'(v) = r$, and
		for every $x \in A$, the distance between $v$ and $x$ is at most 2 in $G_{t'}[R]$.
	If $\tau(A)$ is a valid tuple for $t$ that is forget-$v$-compatible with $\gamma(A')$, then $R$ is a candidate solution for $t$ with respect to $\tau(A)$.
	\end{nclaim}
	
	\begin{proof}
	Let $t$ be a forget node with child $t'$ such that $v \in X_{t'} \setminus X_{t}$, $A \subseteq X_t$ and $A' = A \cup \{v\}$.
	Furthermore, our assumption includes that $\gamma(A')$ is a valid tuple for $t'$ and $R$ is a candidate solution for $t'$ with respect to $\gamma(A')$ such that $v \in R$, $f'(v) = r$ and for every $x \in A$, the distance between $x$ and $v$ is at most 2 in $G_{t'}[R]$.
	
	Since $R$ is a candidate solution for $t'$ with respect to $\gamma(A')$, every $y \in R \setminus A'$ is of distance at most 2 with every $x \in R$ in $G_{t'}[R]$.
	As $A \subseteq A'$, and by our assumption, the distance between $v$ and every $x \in A$ is at most $2$, hence, every $y \in R \setminus A$ is of distance at most 2 to every $x \in R$.
	It ensures us that item-(\ref{prop-forgotten-dist-2-part-1}) from Definition \ref{defn:tw-dp-tuple-candidate-soln} is satisfied.
	Consider every $x \in R \setminus A$.
	Observe that $R \setminus A = (R \setminus A') \cup \{v\}$.
	As $f(v) = r$, hence $v$ is part of at least $r$ triangles in $G_{t'}[R]$.
	Given that $G_t[R]$ and $G_{t'}[R]$ are the same graphs, $v$ is part of at least $r$ triangles in $G_t[R]$.
	Consider every $x \in R \setminus A$ other than $v$.
	As $R$ is a candidate solution for $t'$ with respect to $\gamma(A')$, every such $x$ is part of at least $r$ triangles in $G_{t'}[R]$.
	Since $G_{t'}[R]$ and $G_t[R]$ are the same graph, item-(\ref{prop-number-of-triangles-forgotten}) is satisfied for $R$.
	Also, $R \cap X_t = A$, hence item-(\ref{prop-bag-intersection}) is satisfied for $R$.
	
	We define $\tau(A) = (f, \cA)$ as follows.
	For every $u \in A$, we set $f(u) = f'(u)$.
	Hence, item-(\ref{prop-number-of-triangles-intersection}) from Definition \ref{defn:tw-dp-tuple-candidate-soln} is satisfied for $R$ and item-(\ref{remove-cond-1}) is satisfied for $f$.
	As $f(v) = r$ is given by assumption, hence item-(\ref{remove-cond-1}) is satisfied for $f$.
	We check every $x \in R \setminus A'$ one by one and let $B_x = N_{G_{t'}}(x) \cap A'$.
	If $v \notin B_x$, then we put $B_x$ into $\cA$. 
	Due to this, item-(\ref{remove-cond-3}) is satisfied for $f$.
	Otherwise, $v \in B_x$, and we put $B_x \setminus \{v\}$ into $\cA$.
	Because of this step, item-(\ref{remove-cond-2}) is satisfied for $f$.
	We put $N_{G_t}(v) \cap A$ into $\cA$. 
	It ensures us that item-(\ref{remove-cond-4}) is satisfied.
	Since $\cA$ stores the neighborhood of all vertices from $R \setminus A$, hence item-(\ref{prop-subset-information-forgot-vertex}) from Definition \ref{defn:tw-dp-tuple-candidate-soln} is satisfied.
	
	Now, we argue that item-(\ref{remove-cond-5}) as follows.
	Consider any $x \in A$.
	By assumption, the distance between $x$ and $v$ is at most 2 in $G_{t'}[R]$.
	Hence, $vx \in E(G)$ or there is a common neighbor $w \in R$ between $v$ and $x$.
	If $vx \in E(G)$, then we are done.
	In case $w \in A$, then $w \in A'$ and $w \neq v$, and there is a path of length at most 2 between $v$ and $x$ in $G_{t'}[A]$ as $w \in A$.
	Otherwise, $w \in R \setminus A'$.
	Then, $B_w \in N_{G_{t'}}(w) \cap A' \in \cA'$.
	Observe that $v, x \in B_w$.
	Then, $B_w \in \cA'$ such that $v, x \in B_w$.
	This completes the justification of item-(\ref{remove-cond-5}).
	
	As $\tau(A)$ satisfies every conditions from Definition \ref{defn:forget-A'-compatible} and $R$ satisfies every condition from Definition \ref{defn:tw-dp-tuple-candidate-soln} for being a solution of $t$ with respect to $\tau(A)$, it completes the proof. 
	\end{proof}
	
	Let $t$ be a forget node, $A \subseteq X_t$ and $\tau(A)$ be a valid tuple for $t$.
	Then, there is a candidate optimal solution $R \subseteq V(G_t)$ for $t$ with respect to $\tau(A)$.
	Hence, $|R| = {\dptw}[t, A, \tau(A)]$.
	If $v \in R$, then it follows from Claim \ref{claim:forget-v-parent} that there is a valid tuple $\gamma(A') = (f', \cA')$ for $t'$ such that $A' = A \cup \{v\}$, and $R$ is a candidate solution for $t'$ with respect to $\gamma(A')$.
	Then, $|R| \leq {\dptw}[t', A', \gamma(A')] \leq \max\limits_{\gamma(A') \in \ccC_v(\tau(A))} {\dptw}[t', A', \gamma(A')]$.
	
	If $v \notin R$, then clearly, $R$ is a candidate solution for $t'$ with respect to $\tau(A)$.
	Hence, ${\dptw}[t, A, \tau(A)] \leq {\dptw}[t', A, \tau(A)]$.
	
	Therefore, $${\dptw}[t, A, \tau(A)] \leq \max \bigg\{ {\dptw}[t', A, \tau(A)],  \max\limits_{\gamma(A') \in \ccC_v(\tau(A))} {\dptw}[t', A', \gamma(A')]\bigg\}$$
	
	\medskip
	
	Conversely, let $\tau(A)$ be a tuple for $t$ and $R_1$ be a candidate optimal solution for $t'$ with respect to $\tau(A)$.
	Then, $R_1$ is a candidate solution for $t$ with respect to $\tau(A)$.
	Hence, $|R_1| \leq {\dptw}[t, A, \tau(A)]$.
	Let us also assume that $R_2$ be a candidate optimal solution for $t'$ with respect to $\beta(A')$ such that $\beta(A') \in {\ccC}_v(\tau(A))$, $A' = A \cup \{v\}$ and $|R_2| = \max\limits_{\gamma(A') \in \ccC_v(\tau(A))} {\dptw}[t', A', \gamma(A')]$.
	 Then, due to Claim \ref{claim:forget-v-child}, $R_2$ is a candidate solution for $t$ with respect to $\tau(A)$.
	 Then, $|R_2| \leq {\dptw}[t, A, \tau(A)]$.
	 Therefore, it follows that ${\rm RHS} = \max\{|R_1|, |R_2|\} \leq {\dptw}[t, A, \tau(A)]$.
	\end{proof}
	
	\smallskip
	
	\item[(iv)] 
	{\bf Join Node}: Let $t$ be a join node with two children $t_1$ and $t_2$ such that $X_t = X_{t_1} = X_{t_2}$.
	Observe that there are three graphs to consider in join node, $G_t, G_{t_1}$, and $G_{t_2}$.
	If a candidate solution for $G_t$ intersects $X_t$ in $A$, then its restriction in $G_{t_1}$ and $G_{t_2}$ also intersect $X_{t_1}$ and $X_{t_2}$, respectively in $A$.
	We define the following term to succinctly capture those details.
	
	
	\begin{definition}[Join-$A$-Compatible]
	\label{defn:join-A-compatible}
	Let $t$ be a join node with two children $t_1, t_2$ such that $X_t = X_{t_1} = X_{t_2}$ and $A \subseteq X_t$.
	Consider valid tuples $\tau(A) = (f, \cA), \gamma_1(A) = (f_1, \cA_1)$ and $\gamma_2(A) = (f_2, \cA_2)$ for $t, t_1$ and $t_2$ respectively.
	Then, $\tau(A)$ is {\em join-$A$-compatible} with $(\gamma_1(A), \gamma_2(A))$ if
	\begin{enumerate}
		\item\label{join-cond-1} $\cA = \cA_1 \cup \cA_2$, 
		\item\label{join-cond-2} for all $u \in A$, $f(u) = \min\{r, f_1(u) + f_2(u) - |E(N(u) \cap A)|\}$, and
		\item\label{join-cond-3} if there are two sets $P, Q \in \cA$ such that $P \cap Q = \emptyset$, then either $P, Q \in \cA_1$ or $P, Q \in \cA_2$.
	\end{enumerate}	
	\end{definition}
	
	\medskip

	Let $\tau(A)$ be a tuple for a join node $t$ with children $t_1$ and $t_2$.
	Consider $\XX(\tau(A))$ the collection of tuple pairs $\gamma_1(A)$ for $t_1$ and $\gamma_2(A)$ for $t_2$ such that $\tau(A)$ is join-$A$-compatible with $(\gamma_1(A), \gamma_2(A))$.
	We are define our recurrence for the join-node in the next claim.
	
	\begin{claim}
	\label{claim:join-node-recurrence}
	Let $t$ be a join node with two children $t_1$ and $t_2$, $A \subseteq X_t$.
	Then, $${\dptw}[t, A, \tau(A)] = 
	\max\limits_{(\gamma_1(A), \gamma_2(A))\in \XX(\tau(A))} \bigg({\dptw}[t_1, A, \gamma_1(A)] + {\dptw}[t_2, A, \gamma_2(A)]\bigg) - |A|$$
	
	The value of ${\dptw}[t, A, \tau(A)]$ can be computed in $2^{2^{\OO(k)}}r^{\OO(k)}n^{\OO(1)}$-time.
	\end{claim}

	\begin{proof}
	We use Definition \ref{defn:join-A-compatible} and prove the following two claims that will help us stating a correct recurrence for the join node.
	
	\begin{nclaim}
	\label{claim:join-parent-case}
	Let $t$ be a join node with children $t_1, t_2$ and $A \subseteq X_t$.
	Suppose that $R$ is a candidate solution for $t$ with respect to $\tau(A)$, $R_1 = R \cap V(G_{t_1})$ and $R_2 = R \cap V(G_{t_2})$.
	Then, there are valid tuples $\gamma_1(A)$ and $\gamma_2(A)$ for $t_1$ and $t_2$ respectively such that
	\begin{itemize}
		\item $\tau(A)$ is join-$A$-compatible with $(\gamma_1(A), \gamma_2(A))$, and
		\item $R_1$ is a candidate solutions for $t_1$ with respect to $\gamma_1(A)$, and $R_2$ is a candidate solution for $t_2$ with respect to $\gamma_2(A)$ respectively.
	\end{itemize} 	
	\end{nclaim}
	
	\begin{proof}
	Let $t$ be a join node with children $t_1, t_2$ and $A \subseteq X_t$.
	Subsequently, we assume that $R$ is a candidate solution for $t$ with respect to $\tau(A) = (f, \cA)$, $R_1 = R \cap V(G_{t_1})$ and $R_2 = R \cap V(G_{t_2})$.
	Then, for all $x \in R \setminus A$, $x$ is part of at least $r$ triangles in $G_t[R]$ and every $v \in A$ is part of $f(v)$ triangles subject to a maximum of $r$.
	As $G_t[R]$ and $G[R]$ are the same graphs, we omit this subscript.
	Similarly, $G_{t_1}[R_1]$ and $G[R_1]$ (respectively, $G_{t_2}[R_2]$ and $G[R_2]$) are the same graphs.
	So, we drop the subscripts.
	
	We focus on $G[R_1]$ and $G[R_2]$ simultaneously and justify the existence of tuples $\gamma(A) = (f_1, \cA_1)$ for $t_1$ and $\gamma(A) = (f_2, \cA_2)$ for $t_2$ as follows.
	It is clear that $A \subseteq R_1$ and $A \subseteq R_2$.
	We focus on the following cases.
	
	\begin{description}
		\item[Case-(a)] $A = R_1 = R_2$.
		Observe that $\cA = \emptyset$.
		Then, for all $u \in A$, we set $f_1(u) = f_2(u) = f(u)$ and $\cA = \cA_1 = \cA_2$.
		Since $R \setminus A = \emptyset$ and $R$ is a canddiate solution for $t$ with respect to $\tau(A)$, it must be that $\cA = \emptyset$.
		Hence, by our choice of $\gamma_1(A) = (f_1, \cA_1)$ and $\gamma_2(A) = (f_2, \cA_2)$.
		Since $G[R]$ is the same graph as $G[R_1]$, $G[A]$, and $G[R_2]$, for every $u \in A$, the number of triangles is same in each of $G[R_1]$ and $G[R_2]$.
		Hence, item-(\ref{prop-number-of-triangles-intersection}) of Definition \ref{defn:tw-dp-tuple-candidate-soln} is satisfied.
		Clearly, $A = R_1 \cap X_t = R_2 \cap X_t$.
		Hence, item-(\ref{prop-bag-intersection}) is satisfied.
		Since $R_1 \setminus A = R_2 \setminus A = \emptyset$, items-(\ref{prop-forgotten-dist-2-part-1}),(\ref{prop-number-of-triangles-forgotten}), and (\ref{prop-subset-information-forgot-vertex}) from Definition \ref{defn:tw-dp-tuple-candidate-soln} are trivially satisfied.
		Hence, $R_1$ is a candidate solution for $t_1$ with respect to $\gamma_1(A)$ and and $R_2$ is a candidate solution for $t_2$ with respect to $\gamma_2(A)$.
		Conditions described in items-(\ref{join-cond-1}) and (\ref{join-cond-2}) are trivially satisfied.
		Since $\cA = \emptyset$, there are no pair of sets in $\cA$ that are pairwise disjoint.
		Therefore, item-(\ref{join-cond-3}) is vacuously true.
		Hence, $\tau(A)$ is join-$A$-compatible with $(\gamma_1(A), \gamma_2(A))$.
		
		\item[Case-(b)] $A = R_1 \subset R_2$.
		In this case, observe that every $x \in R \setminus A$ is in $R_2 \setminus A$.
		We assign $\cA_1 = \emptyset$ and $\cA_2 = \cA$.
		Define $f_1(u) = |E(N(u) \cap A)|$ and $f_2(u) = f(u)$.
		Clearly observe that items-(\ref{join-cond-1}) and (\ref{join-cond-2}) are satisfied.
		Consider two sets $X, Y \in \cA$ such that $X \cap Y = \emptyset$.
		As $\cA = \cA_2$, clearly $X, Y \in \cA_2$.
		Since $\cA_1 = \emptyset$, $X, Y \notin \cA_1$.
		Hence, item-(\ref{join-cond-3}) is satisfied.
		Then, $\tau(A)$ is join-$A$-compatible with $((f_1, \cA_1), (f_2, \cA_2))$.
		
		It is not hard to observe that every vertex $u$ of $R_1$ is in $A$ and each of them are part of $f_1(u)$ triangles only, hence satisfying items-(\ref{prop-bag-intersection}) and (\ref{prop-number-of-triangles-intersection}) for $R_1$.
		Since $R_1 \setminus A = \emptyset$, clearly, $\cA_1 = \emptyset$ trivially satisfies items-(\ref{prop-forgotten-dist-2-part-1}),(\ref{prop-number-of-triangles-forgotten}), and (\ref{prop-subset-information-forgot-vertex}) for $R_1$.
		Hence, $R_1$ is a candidate solution for $t_1$ with respect to $\gamma_1(A) = (f_1, \cA_1)$.
		
		Concerning $R_2$, observe that $R_2 \cap X_t = A$.
		Hence, item-(\ref{prop-bag-intersection}) is satisfied for $R_2$.
		Consider the number of triangles that $u \in A$ is part of in $G[R_2]$.
		Since $R_2 = R$, any triangle that $u \in A$ is part of is in $G[R_2]$.
		Hence, $f_2(u) = f(u)$ correctly satisfies item-(\ref{prop-number-of-triangles-intersection}) for $R_2$.
		Consider any $x \in R_2 \setminus A$.
		Then, $x \in R \setminus A$.
		As $R$ is a candidate solution for $t$ with respect to $\tau(A)$, every $x \in R \setminus A$ is in distance at most 2 to every $y \in R$ in $G[R]$.
		Since $G[R]$ and $G[R_2]$ are the same graphs, every $x \in R_1 \setminus A$ is in distance at most 2 to every $y \in R$ in $G[R_2]$.
		This satisfies item-(\ref{prop-forgotten-dist-2-part-1}).
		Next, observe that if $x \in R \setminus A$ is part of a triangle $T$ in $G[R]$, then that triangle is also part of $G[R_2]$ and vice versa.
		Since $x \in R_2 \setminus A$ is part of at least $r$ triangles in $G[R]$, $x$ is also part of at least $r$ triangles in $G[R_2]$.
		It ensures item-(\ref{prop-number-of-triangles-forgotten}).
		Finally, consider any $B \in \cA_2$.
		Note that $\cA_2 = \cA$ and there is $x \in R \setminus A$ such that $N(x) \cap A = B$.
		As $R_2 = R$, for every $B \in \cA_2$, there is $x \in R_2 \setminus A$ such that $N(x) \cap A = B$, satisfying item-(\ref{prop-subset-information-forgot-vertex}).
		Hence, $R_2$ is a candidate solution for $t_2$ with respect to $\gamma_2(A) = (f_2, \cA_2)$.
				
		\item[Case-(c)] $A = R_2 \subset R_1$.
		In this case, observe that every $x \in R \setminus A$ is in $R_1 \setminus A$.
		We assign $\cA_1 = \cA$ and $\cA_2 = \emptyset$.
		Define $f_1(u) = f(u)$ and $f_2(u) = |E(N(u) \cap A)|$.
		The arguments are symmetric to the Case-(b) above.
		
		\item[Case-(d)] $A \subset R_1$ and $A \subset R_2$.
		Observe that $R_2 \cap X_t = R_1 \cap X_t$, hence satisfying item-(\ref{prop-bag-intersection}) for both $R_1$ and $R_2$.
		
		We construct $\gamma_1(A) = (f_1, \cA_1)$ and $\gamma_2(A) = (f_2, \cA_2)$ as follows.
		
		We define $f_1: A \rightarrow [r] \cup \{0\}$ and $f_2: A \rightarrow [r] \cup \{0\}$ as follows.
		For every $u \in A$, we focus on the number of triangles in $G[R_1]$ that $u$ is part of.
		We assign that value (subject to a maximum of $r$) to $f_1(u)$.
		Similarly, we consider the number of triangles in $G[R_2]$ that $u$ is part of.
		We assign that value (subject to a maximum of $r$) to $f_2(u)$.
		Clearly, both $f_1$ and $f_2$ satisfy item-(\ref{prop-number-of-triangles-intersection}) for $R_1$ and $R_2$, respectively.
		
		Let us justify that item-(\ref{join-cond-2}) is satisfied for $f_1$ and $f_2$.
		Consider any $v \in A$ and any arbitrary triangle $T_v$ in $G[R]$ that $v$ is part of.
		If $T_v \subseteq A$, then $T_v$ is part of both $G[R_1]$ and $G[R_2]$.	
		The number of such triangles $T_v \subseteq A$ is precisely the same as the number of edges in  $E(N(v) \cap A)$.
		Hence, number of triangles $T_v$ with $T_v \subseteq A$, i.e. $|E(N(v) \cap A)|$ is counted both in $f_1(v)$ and $f_2(v)$.
		
		If $T_v \cap (R_1 \setminus A) \neq \emptyset$, then observe that $T_v$ contains a vertex $x$ that is not adjacent to any vertex in $R_2 \setminus A$.
		Hence, $T_v$ is part of $G[R_1]$ but not part of $G[R_2]$, hence only counted in $f_1(v)$ but not in $f_2(v)$.
		Similarly, if $T_v \cap (R_2 \setminus A) \neq \emptyset$, then observe that $T_v$ is part of $G[R_2]$ but not $G[R_1]$.
		Hence, $T_v$ that is only part of $G[R_2]$ but not of $G[R_1]$ is only counted in $f_2(v)$.
		Therefore, $f(v) = f_1(v) + f_2(v) - |E(N(v) \cap A)|$ subject to a maximum of $r$, ensuring us that item-(\ref{join-cond-2}) is satisfied.
		
		Now, we decide $\cA_1$ and $\cA_2$.
		Note that if a set $B \in \cA$ then there exists $w \in R \setminus A$ such that $N(w) \cap A = B$.
		%
		We consider every $x \in R_1 \setminus A$ and every $y \in R_2 \setminus A$ one by one.
		Note that both $B_x = N_{G_t}(x) \cap A$ and $B_y = N_{G_t}(y) \cap A$ are in $\cA$.
		We put $B_x$ into $\cA_1$ and $B_y$ into $\cA_2$.
		Hence, every set of $\cA$ is in $\cA_1$ or in $\cA_2$, or both.
		Therefore, $\cA = \cA_1 \cup \cA_2$, satisfying item-(\ref{join-cond-1}).
		
		Consider $P, Q \in \cA$ such that $P \cap Q = \emptyset$.
		Then, observe that if $w, z \in R \setminus A$ such that $N(w) \cap A = P$ and $N(z) \cap A = Q$, then it cannot be possible that $w \in R_1 \setminus A$ and $z \in R_2 \setminus A$.
		If $w \in R_1 \setminus A$ and $z \in R_2 \setminus A$, then $w$ and $z$ have no common neighbor in $R$ and they cannot be at distance 2 with each other.
		This contradicts that $R$ is a candidate solution for $t$ with respect to $\tau(A)$.
		Hence, it must be that both $w, z \in R_1 \setminus A$ or both $w, z \in R_2 \setminus A$.
		It implies that either $P, Q \in \cA_1$ or $P, Q \in \cA_2$, completing the justification of item-(\ref{join-cond-3}).
		We have justified all conditions of Definition \ref{defn:join-A-compatible}, implying that $\tau(A)$ is join-$A$-compatible with $(\gamma_1(A), \gamma_2(A))$. 
		
		By construction of $\cA_1$ and $\cA_2$, note that $B \in \cA_1$ (respectively, $B \in \cA_2$) if and only if there is $x \in R_1 \setminus A$ (respectively, $x \in R_2 \setminus A$) such that $N(x) \cap A = B$.
		Hence, item-(\ref{prop-subset-information-forgot-vertex}) are satisfied.
		
		Now, we consider any vertex $x \in R_1 \setminus A$.
		Since $x \in R \setminus A$, $x$ particpates in at least $r$ triangles in $G[R]$.
		Each such triangle are in $G[R_1]$, since $x$ has no neighbor in $R_2 \setminus A$.
		Hence, $x \in R_1 \setminus A$ is part of at least $r$ triangles in $G[R_1]$.
		Similarly, we can prove that any $x \in R_2 \setminus A$ is part of at least $r$ triangles in $G[R_2]$ ensuring item-(\ref{prop-number-of-triangles-forgotten}) of Definition \ref{defn:tw-dp-tuple-candidate-soln}.
		
		Consider any $x \in R_1 \setminus A$ and $y \in R_1$.
		Since $x \in R \setminus A$ and $y \in R$, the distance between $x$ and $y$ is at most 2 in $G[R]$.
		But given that $x \in R_1 \setminus A$, hence $x$ cannot have any neighbor in $R_2 \setminus A$.
		And there is a common neighbor $z$ between $x$ and $y$, that is not in $R_2 \setminus A$, implying that $z$ must be in $R_1$.
		Hence, the distance between $x$ and $y$ is at most 2 in $G[R_1]$.
		By similar arguments, we can prove that between any $x \in R_2 \setminus A$ and $y \in R_2$, the distance is at most 2 in $G[R_2]$.
		This ensures that item-(\ref{prop-forgotten-dist-2-part-1}) is satisfied for $R_1$ and $R_2$ each.
		
		As we have established all conditions of Definition \ref{defn:tw-dp-tuple-candidate-soln}, $R_1$ is a candidate solution for $t_1$ with respect to $\gamma_1(A)$ and $R_2$ is a candidate solution for $t_2$ with respect to $\gamma_2(A)$.
	\end{description}
	
	As the cases are mutually exhaustive, this concludes the proof.
	\end{proof}
	
	\medskip
	
	Now, we prove the converse of the above claim.
	
	\begin{nclaim}
	\label{claim:join-chid-case}
	Let $t$ be a join node with children $t_1, t_2$ and $A \subseteq X_t$.
	Consider tuple $\tau(A)$, $\gamma_1(A)$ and $\gamma_2(A)$ for $t, t_1$ and $t_2$ respectively such that $\tau(A)$ is join-$A$-compatible with $(\gamma_1(A), \gamma_2(A))$.
	Suppose that $R_1$ is a candidate solution for $t_1$ with respect to $\gamma_1(A)$ and $R_2$ is a candidate solution $t_2$ with respect to $\gamma_2(A)$.
	Then, $R_1 \cup R_2$ is a candidate solution for $t$ with respect to $\tau(A)$.
	\end{nclaim}
	
	\begin{proof}
	Suppose that the premise of the statement is true.
	We consider $R = R_1 \cup R_2$.
	Moreover, $R_1$ is a candidate solution for $t_1$ with respect to $\gamma_1(A) = (f_1, \cA_1)$ and $R_2$ is a candidate solution for $t_2$ with respect to $\gamma_2(A) = (f_2, \cA_1)$.
	As our assumption includes that $\tau(A) = (f, \cA)$ is join-$A$-compatible with $(\gamma_1(A), \gamma_2(A))$, it must be that for every $u \in A$, $f(u) = \min\{r, f_1(u) + f_2(u) + |E(N(u) \cap A)|$.
	Moreover, $\cA = \cA_1 \cup \cA_2$ and if there are two disjoint sets in $\cA$, then either both of them are in $\cA_1$ or both of them are in $\cA_2$.
	
	Concerning the set $R$, observe that $R \cap X_t = A$, hence item-(\ref{prop-bag-intersection}) of Definition \ref{defn:tw-dp-tuple-candidate-soln} is satisfied.
	Consider $x \in R \setminus A$ and $y \in R$.
	If $y \in A$, then there are two subcases, $x \in R_1 \setminus A$ or $x \in R_2 \setminus A$.
	If $x \in R_1 \setminus A$, then due to $R_1$ being a candidate solution for $t_1$ with respect to $\gamma_1(A)$, there is a path of length at most 2 between $x$ and $y$ in $G[R_1]$.
	As every edge of $G[R_1]$ is an edge of $G[R]$, therefore, the distance between $x$ and $y$ is at most 2 in this case.
	Similarly, if $x \in R_2 \setminus A$ and $y \in A$, then the argument is symmetric.
	In case, $x \in R_1 \setminus A$ and $y \in R_2 \setminus A$, then we note $B_x = N(x) \cap A$ and $B_y = N(y) \cap A$.
	Note that $B_x \in \cA_1$ and $B_y \in \cA_2$.
	But this is possible only when $B_x \cap B_y \neq \emptyset$.
	Therefore, $x$ and $y$ have a common neighbor in $A$ in the graph $G[R]$.
	Hence, the distance between $x$ and $y$ is at most 2.
	As we have covered all possible cases, item-(\ref{prop-forgotten-dist-2-part-1}) is satisfied for $R$.
	
	Now, we look at the number of triangles every $u \in A$ is part of.
	There are 3 types of triangles.
	First type are the ones when all vertices including $u$ of the triangles are in $A$.
	The number of such triangles is $|E(N(u) \cap A)|$.
	The second type is the number of triangles in $G[R_1]$ that $u$ is part of.
	It is $f_1(u)$ and it includes the number of triangles in with all endpoints in $A$.
	The third ytpe is the number of triangles in $G[R_2]$ that $u$ is part of.
	It is $f_2(u)$ that also includes the number of triangles with all 3 vertices in $A$.
	Hence, $|E(N(u) \cap A)|$ is counted twice, once in $f_1(u)$ and once in $f_2(u)$.
	As $f(u)$ is set to $f_1(u) + f_2(u) + |E(N(u) \cap A)|$ subject to a maximum of $r$, every $u \in A$ is part of $f(u)$ triangles in $G[R]$ that is correct, satisfying item-(\ref{prop-number-of-triangles-intersection}).
	
	Consider any triangle in $G[R]$ that $x \in R \setminus A$ is part of.
	If $x \in R_1 \setminus A$, then such a triangle is in $G[R_1]$, and there are $r$ such triangles in $G[R_1]$.
	Similarly, if $x \in R_2 \setminus A$, then such a triangle is in $G[R_2]$, and there are $r$ such triangles in $G[R_2]$.
	Since $G[R_1]$ and $G[R_2]$ both are subgraphs of $G[R]$, every $x \in R \setminus A$ is part of at least $r$ triangles in $G[R]$, satisfying item-(\ref{prop-number-of-triangles-forgotten}).
	
	Finally, consider any $B \in \cA$.
	Since $\cA = \cA_1 \cup \cA_2$, it must be that $B \in \cA_1$ or $B \in \cA_2$.
	If $B \in \cA_1$, then there is $x \in R_1 \setminus A$ such that $N(x) \cap A = B$.
	The same vertex $x \in R \setminus A$.
	Similarly, if $B \in \cA_2$, we can argue that there is $x \in R \setminus A$ such that $N(x) \cap A = B \in \cA$.
	Hence, item-(\ref{prop-subset-information-forgot-vertex}) is satisfied.
	As $R$ satisfies all conditions of Definition \ref{defn:tw-dp-tuple-candidate-soln}, this completes the proof.
	\end{proof}
	
	\medskip
	
	Let $t$ be a join node with two children $t_1$ and $t_2$, $A \subseteq X_t$.
	Considering the tuple $\tau(A)$ for $t$, in order to compute ${\dptw}[t, A, \tau(A)]$, first we enumerate all possible functions $f_1, f_2: A \rightarrow [r] \cup \{0\}$ such that for every $u \in A$, $f(u) = \min\{r, f_1(u) + f_2(u) + |E(N(u) \cap A)|\}$ is satisfied.
	Observe that there are $(r+1)^{2|A|}$ possible such combinations of functions that need to be checked for this.
	Subsequently, we enumerate $\cA_1, \cA_2 \subseteq \cA$ such that $\cA_1 \cup \cA_2 = \cA$.
	There are $4^{|\cA|}$ such choices of $(\cA_1, \cA_2)$.
	Now, for each such combination $\gamma_1(A) = (f_1, \cA_1)$ and $\gamma_2(A) = (f_2, \cA_2)$, we check whether $\tau(A)$ is join-$A$-compatible with $(\gamma_1(A), \gamma_2(A))$ or not.
	This can be checked in polynomial-time.
	Subsequently, after each such join-$A$-compatible tuple pair $(\gamma_1(A), \gamma_2(A))$ is enumerate, we calculate the values accordingly.
	Since $4^{|\cA|} \leq 2^{2^{\OO(k)}}$, that this procedure takes $2^{2^{\OO(k)}}r^{\OO(k)}n^{\OO(1)}$-time.
	
	Now, we establish the correctness of the recurrence.
	Consider any optimal candidate solution $R$ of $t$ with respect to $\tau(A)$.
	Due to Claim \ref{claim:join-parent-case}, there are tuples $\gamma_1(A)$ for $t_1$ and $\gamma_2(A)$ for $t_2$ such that $\tau(A)$ is join-$A$-compatible with $(\gamma_1(A), \gamma_2(A))$.
	Furthermore, $R_1$ is a candidate solution for $t_1$ with respect to $\gamma_1(A)$, and $R_2$ is a candidate solution for $t_2$ with respect to $\gamma_2(A)$.
	Then, $|R_1| \leq {\dptw}[t_1, A, \gamma_1(A)]$ and $|R_2| \leq {\dptw}[t_2, A, \gamma_2(A)]$.
	As $A = R_1 \cap R_2$, therefore $${\dptw}[t, A, \tau(A)]$$
	$$\leq \max\limits_{(\gamma_1(A), \gamma_2(A))\in \XX(\tau(A))} \bigg({\dptw}[t_1, A, \gamma_1(A)] + {\dptw}[t_2, A, \gamma_2(A)]\bigg) - |A|$$
	
	Conversely, let $R_1$ is a candidate optimal solution for $t_1$ with respect to $\gamma_1(A)$, and $R_2$ be a candidate optimal solution for $t_2$ with respect to $\gamma_2(A)$.
	Moreover, $\tau(A)$ is join-$A$-compatible with $(\gamma_1(A), \gamma_2(A))$.
	Then, due to Claim \ref{claim:join-chid-case}, $R$ is a candidate solution for $t$ with respect to $\tau(A)$ and $|R| = |R_1| + |R_2| - |A|$.
	As $|R| \leq {\dptw}[t, A, \tau(A)]$, it follows that
	$${\dptw}[t, A, \tau(A)]$$
	$$\geq \max\limits_{(\gamma_1(A), \gamma_2(A))\in \XX(\tau(A))} \bigg({\dptw}[t_1, A, \gamma_1(A)] + {\dptw}[t_2, A, \gamma_2(A)]\bigg) - |A|$$
	
	Hence, this recurrence is correct.
	\end{proof}
	
	We have established the recurrence for each of the node types and their correctness.
\end{description}
Observe from Claims \ref{claim:intro-node-correctness}, \ref{claim:correctness-forget-node}, and \ref{claim:join-node-recurrence} that computation of every table entry takes at most $2^{2^{\OO(k)}}r^{\OO(k)}n^{\OO(1)}$-time.
The number of table entries in every specific node is at most $2^{k}2^{2^{\OO(k)}}r^{\OO(k)}$.
Hence the collection of all table entries can be computed in $2^{2^{\OO(k)}}r^{\OO(k)}n^{\OO(1)}$-time.
\end{proof}

We use 
the following proposition from \cite{HartungKN15} to prove Theorem \ref{thm:triangle-tw-dp}.

\begin{proposition}[Lemma 5 of \cite{HartungKN15}]
\label{prop:subset-tree-decomposition-property}
Let $G = (V, E)$ be an undirected graph and $S \subseteq V(G)$.
Then, for any tree decomposition $\cT = (T, \{X_t\}_{t \in V(T)})$ of $G$, there is at least one vertex $v \in S$ such that there is a bag that contains $N[v] \cap S$.
\end{proposition}


{\ThmTwDPTriangle*}

\begin{proof}
The first thing to observe is that every vertex is part of at least $r$ triangles in any optimal solution.
Hence, 
we delete all vertices of $G$ that are part of less than $r$ triangles in $G$, and hence assume that every vertex of $G$ is part of at least $r$ triangles in $G$.
As treewidth is minor-closed, we invoke the algorithm by Korhonen \cite{Korhonen21} that runs in $2^{\OO({\trw})}n^{\OO(1)}$-time and computes a tree decomposition of width $k \leq 2{\trw} + 1$.
We convert it into a nice tree decomposition $\cT = (T, \{X_t\}_{t \in V(T)})$ of with $k$ in polynomial-time (see \cite{Kloks94}). 
Fix a hypothetical $2$-club $S$ of maximum cardinality that satisfies {\vtrangle}.
Due to Proposition \ref{prop:subset-tree-decomposition-property}, there is $w \in S$ such that $N_G[w] \cap S \subseteq X$ for a bag $X$ in $\cT$.
We consider $A_w = N_G[w] \cap S \subseteq X$.
Since the number of bags in $\cT$ is $\OO(n)$, we can first guess a bag $X_p$ (corresponding to node $p \in V(T)$), a vertex $w$ part of the potential maximum sized 2-club with {\vtrangle}, and $A_w$, the neighbors of $w$ in $X_p$.
After having done this, we set the node $p$ as the root of $\cT$.
Furthermore, we may assume that $A_w = X_p$.
In case $A_w \subset X_p$, we can create a sequence of forget nodes that starts in $X_p$ and deletes all vertices in $X_p \setminus A_w$ one by one.

Next, we consider the nice tree decomposition $\cT = (T, \{X_t\}_{t \in V(T)})$ with root $p$, and invoke Lemma \ref{lemma:treewidh-DP-table} and compute ${\dptw}[p, X_p, \tau(X_p)]$.
We choose $\tau(X_p) = (f, \cA)$ such that for all $u \in X_p$, $f(u) = r$, and for every $x, y \in X_p$, $xy \in E(G)$ or $x, y \in B$ for some $B \in \cA$.
This ensures us that we compute a potential solution of largest size that contains $X_p$.
The number of subproblems in every bag is $2^{2^{\OO(k)}}r^{\OO(k)}$, and due to Lemma \ref{lemma:treewidh-DP-table}, computing the subproblems take $2^{2^{\OO(k)}}r^{\OO(k)}n^{\OO(1)}$-time, the algorithm runs in $2^{2^{\OO(k)}}r^{\OO(k)}n^{\OO(1)}$-time.
Since $k \leq 2{\trw} + 1$, for any potential 2-club $S$ of $G$ satisfying {\vtrangle}, the treewidth of $G[S]$ is at most $2{\trw} + 1$.
Note that the degeneracy of a graph is bounded by its treewidth.
Thus, the degeneracy of $G[S]$ is at most  $2{\trw} + 1$.
Therefore, in $G[S]$, there is a vertex $u$ with degree at most $2{\trw} + 1$.
Then, $u$ is part of at most ${{2{\trw} + 1}\choose{2}}$ triangles in $G[S]$.
It implies that $r \leq {{2{\trw} + 1}\choose{2}}$ which is $\OO({\trw}^2)$.
As $r$ is $\OO({\trw}^2)$ and $k$ is $\OO({\trw})$, it follows that the algorithm runs in $2^{2^{\OO({\trw})}} {\trw}^{\OO({\trw})} n^{\OO(1)}$-time.
Therefore, the claimed running time follows.
\end{proof}

%
%
%

\section{Parameterized by the feedback edge number and $h$-index}
\label{sec:fes-result}

In this section, we provide our results when parameterized by the feedback edge number (${\fes}$) as well as the (${\hind}$) of the input graph.

\subsection{Polynomial kernel parameterized by the feedback edge number}

Given an undirected graph $G$, we can compute a minimum feedback edge set $F \subseteq E(G)$ in $\OO(m+n)$-time by computing a spanning forest of the graph.
Consider an instance $(G, \ell)$ of {\trangle} {\twoclub} problem such that $F \subseteq E(G)$ is a minimum feedback edge set of $G$ of size $k$.
Let $T$ be the corresponding spanning forest of $G$ such that no edge of $F$ is in $T$.
Note that the number of endpoints in $F$ is at most $2{\fes}(G)$.
Let $k = {\fes}(G)$ be the parameter, $D$ denote the endpoints of the edges of $F$, and $Y = V(G) \setminus D$.
Therefore, $|D| \leq 2{\fes}(G) = 2k$. We now provide a linear kernel for {\trangle} {\twoclub} when parameterized by ${\fes}$, the feedback edge number of the input graph.
In-fact, we prove something stronger, that {\trangle} {\sclub} admits a linear kernel when parameterized by ${\fes}$.


Since any cycle of $G$ intersects at least one edge of $F$, for any vertex $v \in Y$, if $v$ is part of a triangle, then that triangle must contain at least one edge from $F$.
We prove the following structural lemma that is useful to prove our result.

\begin{lemma}
\label{lemma:fes-vertex-triangle-part}
Let $(G, \ell)$ be an instance of {\trangle} {\sclub} and $D$ denote the endpoints of the edges present in a minimum feedback edge set.
Then, the following statements hold true.
\begin{enumerate}[(i)]
	\item\label{fes-vertex-no-triangle} any vertex $v \in V(G) \setminus D$ is part of at most ${\fes}(G)$ triangles in $G$, and
	\item\label{fes-triangle-s-club-relation} if $r > {\fes}(G)$, then every vertex $r$-triangle $s$-club $S$ of $G$ has at most $2{\fes}(G)$ vertices and $S \subseteq D$. 
\end{enumerate} 
\end{lemma}

\begin{proof}
Consider an instance $(G, \ell)$ of {\trangle} {\sclub} such that $k = {\fes}(G)$ and $D$ be the endpoints of a minimum feedback edge set.
We prove the items in the given order.
\begin{enumerate}[(i)]
	\item As $F$ is a (minimum) feedback edge set of $G$, any cycle of $G$ intersects at least one edge in $F$.
	Therefore, any triangle of $G$ is intersected by $F$.
	Then, any $v \in V(G) \setminus D$ that is part of a triangle must use at least one edge from $F$.
	If $v$ is part of two distinct triangles $B_1, B_2$ of $G$, then $B_1$ and $B_2$ can have at most one edge $xy$ in common.
	Since $v$ is part of both $B_1$ and $B_2$, the common edge must also contain $v$.
	As $v$ is not an endpoint of any edge of $F$, the edge of $F$ present in $B_1$ and the edge of $F$ present in $B_2$ are different.
	Therefore, there does not exist any triangle of $G$ that contains $v$ and uses two edges of $F$.
	Hence, every $v \in V(G) \setminus D$ is part of at most ${\fes}(G)$ triangles, completing the proof of the first part of the statement.

	\item
	To prove the second part, we use the first item.
	From item (\ref{fes-vertex-no-triangle}), we have that every vertex of $V(G) \setminus D$ is part of at most $r$ triangles in $G$.
	If $r > {\fes}(G)$, and $S$ be a vertex $r$-triangle $s$-club of $G$, then it is clear that $S$ cannot contain any vertex from $V(G) \setminus D$.
	Therefore, $S \subseteq D$, implying that $|S| \leq |D| \leq 2{\fes}(G)$.
\end{enumerate}
This completes the proof of the lemma.
\end{proof}

We can assume without loss of generality that every vertex is part of at least one triangle.
Consider a feedback edge $uv \in F$.
As $G - F$ is a forest, there is a unique path $P_{u, v}$ between $u$ and $v$ in $G - F$;
we call this path 
a {\em feedback edge path} between $u$ and $v$ with respect to $F$.
If $w$ is in the path $P_{u, v}$, then $uv$ is a {\em spanning feedback edge} of $w$.
If $w$ is in the path $P_{u, v}$ and $u, v \in N_G(w)$, then $w$ is {\em satisfied} by the feedback edge $uv$.
Equivalently, we say that the feedback edge $uv$ {\em satisfies} the vertex $w$.

\begin{lemma}
\label{lemma:spanning-feedback-edge-triangle-constraint}
Suppose that every vertex of $G$ is part of a triangle. Let $F$ be a (minimum) feedback edge set of $G$, and $D = V(F)$.
For every $w \in V(G) \setminus D$ and every feedback edge $uv \in F$, $w$ is satisfied by a feedback edge $uv$ if and only if $\{w, u, v\}$ induces a triangle in $G$.
\end{lemma}

\begin{proof}
Since $F$ is a feedback edge set of $G$, there is a unique path between every pair of vertices of $F$ in $G-F$.

For the backward direction ($\Leftarrow$), let $\{w, u, v\}$ induces a triangle in $G$.
As $uv \in F$, in the graph $G - F$, there is a unique path $P_{u, v}$ between $u$ and $v$ in $G - F$.
Since $w \notin D$, $wu, wv \notin F$.
Then, $(u, w, v)$ is a path between $u$ and $v$ in $G - F$, and due to the uniqueness of $P_{u, v}$, it must be that $P_{u,v} = (u, w, v)$ itself.
Thus, $w$ is satisfied by the feedback edge $uv$ and $u, v \in N_G(w)$.

For the forward direction ($\Rightarrow$), suppose that $uv$ is a feedback edge of $G$, and $w \in V(G) \setminus D$ is satisfied by $uv$.
Then, $w$ is in the unique path $P_{u, v}$ between $u$ and $v$ in $G - F$ and $u, v \in N_G(w)$.
Concerning $\{u, v, w\}$, observe that $uv, vw, uw \in E(G)$.
Hence $\{u, v, w\}$ forms a triangle of $G$, this completes the proof of the lemma.
\end{proof}

The following observation holds true due to the property of a feedback edge set.

\begin{observation}
\label{obs:feedback-edge-satisfies-unique-vertex}
Let $F$ be a feedback edge set of $G$, $uv \in F$, and $D = V(F)$.
Then, there is at most one vertex $w \in V(G) \setminus D$ that is satisfied by $uv$.
\end{observation}

\begin{proof}
Let $F$ be a feedback edge set of $G$, $uv \in F$, and $D = V(F)$.
It suffices to prove that in case there is $w \in V(G) \setminus D$ that is satisfied by $uv$, then $w$ is unique.
Consider any $w \in V(G) \setminus D$.
If there are two distinct vertices $w, w' \in V(G) \setminus D$ such that both $w$ and $w'$ are satisfied by the feedback edge $uv$, then due to Lemma \ref{lemma:spanning-feedback-edge-triangle-constraint}, $\{u, v, w\}$ and $\{u, v, w'\}$ in $G$.
Then the vertices $u, v, w, w'$ form a cycle $C_4$ as a subgraph in $G - F$ without the feedback edge $uv$.
This is a contradiction to the assumption that $F$ is a feedback edge set of $G$.
Hence, the statement is true.
\end{proof}

Using the above observation and lemmas, we can prove the main result of this section.

{\ThmFESTriangle*}

\begin{proof}
Let $(G, \ell)$ be an instance of a {\trangle} {\twoclub} in which we have exhaustively deleted every vertex of the graph that does not participate in any triangle.
Let $F \subseteq E(G)$ be a minimum feedback edge set of $G$ and $D = V(F)$.
We use $Y = V(G) \setminus D$.
Observe that $|D| \leq 2{\fes}$.
If we can prove that $|Y| \leq {\fes} + 1$, then we are done.

\begin{description}
	\item[Case (i):] If $r > {\fes}$, then due to the item-(\ref{fes-triangle-s-club-relation}) of Lemma \ref{lemma:fes-vertex-triangle-part}, every vertex $r$-triangle $s$-club $S$ is contained in $D$.
	Hence, if $\ell > |D|$, it is clear that $(G, \ell)$ is a no-instance.
	In such a situation, we return a trivial no-instance which is a diamond with $\ell = 4$ and $r = 3$.
	Otherwise, $\ell \leq |D|$, and we return $(G[D],\ell)$ as the output instance.
	Let us first justify that $(G[D], \ell)$ is a yes-instance if and only if $(G, \ell)$ is a yes-instance.
	
	For the forward direction ($\Rightarrow$), let $(G, \ell)$ be a yes-instance.
	Due to item-(\ref{fes-triangle-s-club-relation}) of Lemma \ref{lemma:fes-vertex-triangle-part}, every vertex $r$-triangle $s$-club is contained in $D$.
	As $(G, \ell)$ is a yes-instance, it must be that $\ell \leq |D|$ and a solution is a subset of $D$.
	Hence, $(G[D], \ell)$ is a yes-instance.
	
	For the backward direction ($\Leftarrow$), observe that any vertex $r$-triangle $s$-club of $(G[D], \ell)$ is a vertex $r$-triangle $s$-club of $(G, \ell)$, since $G[D]$ is an induced subgraph of $G$.
	Hence, the backward direction ($\Rightarrow$) follows.
	
	Since $|D| \leq 2{\fes}$, our output instance has at most $2{\fes}$ vertices.
	As deleting ${\fes}$ many edges from $G[D]$ removes all cycles from $G[D]$, it must be that the number of edges of $G[D]$ is at most $|D| - 1 + {\fes}$ which is at most $3{\fes} - 1$.
	
	\item[Case (ii):] If $r \leq {\fes}(G)$, considering $k = {\fes}(G)$, we prove by constructing a sequence sets and subgraphs such that there exists a sequence $Y_0 \subseteq Y_1 \subseteq \cdots \subseteq Y_k$ of vertex subsets of $Y$, and a sequence of subgraphs $T_0, T_1,\ldots,T_k$ of $G$ such that $T_0 = G - F$, $Y_0 = \emptyset$, and the following properties are satisfied: 
	\begin{itemize}
		\item for every $i \in \{0,1,\ldots,k\}$, $|Y_i| \leq i$, and $T_i$ is constructed by adding exactly $i$ edges from $F$ into $T_0$, 
		\item for every $i \in \{0,1,\ldots,k-1\}$, $T_{i+1}$ is constructed from $T_i$ by adding exactly one edge from $F$ that is not already in $T_{i}$,
		\item for every $i \in \{0,1,\ldots,k-1\}$, if there is $w \in Y_{i+1} \setminus Y_{i}$, then $w$ is satisfied by the edge $uv \in T_{i+1} \setminus T_i$.
	\end{itemize}
	
	{\bf Base case:} $i = 0$.
	Since $T_0 = G - F$, $Y_0 = \emptyset$ and each of the three items follow.
	The starting step is to start with $T_0 = G - F$ and $Y_0 = \emptyset$.
	
	{\bf Induction step:} Starting from $i = 0,\ldots,k-1$ in this sequence, we construct $T_{i+1}$ from $T_{i}$.
	By {\em induction hypothesis}, we assume that these three properties of this claim hold true for every $i \leq j$.
	
	Now, we consider $i = j+1$.
	Due to the induction hypothesis, for every $0 \leq i \leq j$, $T_i$ is constructed from $T_0$ by adding exactly $i$ edges of $F$ into $T_0$.
	Furthermore, due to the induction hypothesis, for every $1 \leq i \leq j$, $|Y_i| \leq i$, and $T_i$ is constructed from $T_{i-1}$ by adding exactly one edge of $F$ and if $w \in Y_i \setminus Y_{i-1}$, then $w$ is satisfied by the edge $uv \in T_i \setminus T_{i-1}$ with a triangle in $G$.
	
	We consider $i = j+1$. 
	Initialize $Y_{j+1} := Y_j$.
	We choose an edge $u^*v^* \in F \setminus T_j$, and set $T_{j+1} = T_j \cup \{u^* v^*\}$.
	Hence, the second condition is satisfied for $T_{j+1}$.
	Due to Observation \ref{obs:feedback-edge-satisfies-unique-vertex}, there is at most one vertex $w \in V(G) \setminus D$ that is satisfied by $u^*v^*$.
	We can find out in polynomial-time whether there exists a vertex $w \in V(G) \setminus D$ or not that is satisfied by $u^*v^*$.
	If such a vertex $w$ exists and not already in $Y_j$, then we set $Y_{j+1} = Y_j \cup \{w\}$.
	
	Moreover, if $w \in Y$ exists that is satisfied by the feedback edge $u^*v^*$, then due to Observation \ref{obs:feedback-edge-satisfies-unique-vertex}, there does not exist any other vertex $w' \in Y$ that is satisfied by $u^*v^*$ with a triangle in $G$.
	Hence, no other vertex can be added to $Y_{j+1}$.
	Since $|Y_j| \leq j$, it follows that $|Y_{j+1}| \leq j+1$.
	As $T_{j}$ was constructed from $T_0$ by adding exactly $j$ edges, and $u^*v^* \in F \setminus T_j$ is added to construct $T_{j+1}$, it implies that $T_{j+1}$ is constructed from $T_0$ by adding $j+1$ edges.
	Therefore, $|Y_{j+1}| \leq j+1$ and exactly $j+1$ edges are added to construct $T_{j+1}$ from $T_0$ satisfying the first condition.
	
	The second condition is satisfied because $T_{j+1}$ is constructed from $T_j$ by adding exactly one edge $u^*v^* \in F$ not already in $T_j$.
	For the third condition, suppose $w \in Y_{j+1} \setminus Y_j$.
	Then, note that we added $w$ to $Y_{j+1}$ only when $w$ is not already in $Y_{j}$ and $w$ is satisfied by $u^*v^*$.
	Therefore, the third condition is also satisfied.
	
	The above construction ensures us that every vertex of $G$ that is part of a triangle is added into $Y_i \cup D$ and $|Y_k| \leq k = {\fes}$.
	Therefore, $|D \cup Y_k| \leq 3{\fes}$.
	Observe that the graph $G[D \cup Y_k] - F$ is a forest, and has at most $3{\fes}$ vertices.
	Hence, this graph has at most $3{\fes} - 1$ edges.
	Therefore, $G[D \cup Y_k]$ has at most $4{\fes} - 1$ edges.
	
	Finally, it remains to prove that $(G, \ell)$ is a yes-instance if and only if $(G[D \cup Y_k], \ell)$ is a yes-instance.
	Clearly, the backward direction ($\Leftarrow$) is trivial since $G[D \cup Y_k]$ is an induced subgraph of $G$ and a vertex $r$-triangle $s$-club of $G[D \cup Y_k]$ is a vertex $r$-triangle $s$-club of $G$.
	
	For the forward direction ($\Rightarrow$), assume that $S$ is a vertex $r$-triangle $s$-club of size $\ell$ in $G$.
	What is crucial here is that our construction of $Y_k$ ensures that all vertices of $Y$ that is part of some triangle in $G$ is added.
	To see this, note that construction of $T_{j+1}$ from $T_j$ involves adding an edge $xy$ from $F \setminus T_j$.
	Subsequently, if there is $w \in Y \setminus Y_j$ such that $w$ is satisfied by the edge $xy$, then $w$ is added.
	If $w$ added to $Y_{j+1}$, then $w$ is satisfied by the edge $xy$, it follows due to Lemma \ref{lemma:spanning-feedback-edge-triangle-constraint} that $\{x, y, w\}$ is a triangle in $G$.
	
	Let us now argue that a vertex $w \in Y \setminus Y_k$ is not part of any triangle in $G$.
	Consider any $w \in Y \setminus Y_k$.
	If $w$ is part of a triangle in $G$, then there must be $xy \in F$ such that $\{x, y, w\}$ is a triangle in $G$.
	Then, due to Lemma \ref{lemma:fes-vertex-triangle-part}, $w$ is satisfied by $xy \in F$.
	Suppose $xy$ is the edge that was added to construct $T_{h+1}$ from $T_h$.
	In that step, we also checked if there is any vertex from $Y$ that is satisfied by $xy$.
	Since $w$ is one such vertex, it must be that $w$ was either in $Y_h$ or $w$ was added to $Y_h$ to obtain $Y_{h+1}$.
	Hence, $w \in Y_k$, contradicting our assumption.
	Therefore, if a vertex $w \in Y \setminus Y_k$, then $w$ is not part of any triangle in $G$.
	Hence, any vertex $r$-triangle $s$-club of $G$ is contained in $G[D \cup Y_k]$.
	It implies that $G[D \cup Y_k]$ is a yes-instance.
\end{description}

Our case analysis has been mutually exhaustive.
In both the cases $r \leq {\fes}(G)$ or $r > {\fes}(G)$, we have proved that the output instance is equivalent to the input instance.
If $r > {\fes}(G)$, we have proved that the output instance has at most $2{\fes}$ vertices and at most $3{\fes} - 1$ edges.
Otherwise, $r \leq {\fes}(G)$ and we have proved that the output instance has at most $3{\fes}$ vertices and $4{\fes} - 1$ edges.
This completes the proof.
\end{proof}

%
%
%

\subsection{XP algorithm parameterized by the $h$-index of the graph}

Now, we provide the proof of Theorem \ref{thm:h-index-XP-algorithm}.

{\ThmHIndex*}

\begin{proof}
Let $G = (V, E)$ be an undirected graph.
The first step is to find the $h$-index of $G$, and a set of vertices $X$ of maximum possible cardinality such that each vertex of $X$ has degree at least $|X|$.
	This set $X$ can be found in polynomial time as follows: Compute the degree of every vertex of $G$ and sort them in the descending order $d_1, d_2, \dotsc, d_n$. Scan this list and find the largest index $k$ such that $d_k \geq k$. The set $X$ is the top $k$ vertices in the list.
	

	Our next step is to guess a subset $Y \subseteq X$ that intersects our potential solution $S^*$ with $X$.
	There are $2^{|X|}$ such guesses.
	Note that every vertex of $Z = V(G) \setminus X$ has degree at most $k$ in $G$.
	Given a fixed set $Y$, we define a relation $\sim$ on the vertices of $Z$ as follows.
	For $x, y \in Z$, if $N_G(x) \cap Y = N_G(y) \cap Y$, then $x \sim y$.
	It is not very hard to observe that $\sim$ is an equivalence relation.
	Consider the equivalence classes of $Z$ under $\sim$.
	For every equivalence class $T$ of $Z$, we consider $|T| + 1$ choices, every choice either chooses a vertex $v \in T$ to be included in the potential solution $S^*$ or decides that no vertex from $T$ will be included at all.
	
	Let $T_1,\ldots,T_q$ denote the set of  guessed equivalence classes of $Z$ under $\sim$ from which one vertex $v_i \in T_i$ is part of the potential solution.
	An equivalence classes $T$ is in {\em $Y$-conflict} with an equivalence class $T'$ if $N_G(T) \cap N_G(T') \cap Y = \emptyset$.
	Consider an equivalence class $T_i$ that is in $Y$-conflict with some other equivalence class $T_j$ such that $j \neq i$.
	Observe that for both $v_i \in T_i$ and $v_j \in T_j$ to be part of the solution, since $v_i$ and $v_j$ have no common neighbor in $Y$, it must be that the distance between $v_i$ and $v_j$ in $G - X$ is at most 2.
	If we are to pick any additional vertex $u$ from $T_i$ other than $v_i$, then it must be that $u$ must be in distance at most 4 from $v_i$.
	
	Let $J \subseteq [q]$ to be the indices such that $T_j$ is in $Y$-conflict with some other equivalence class $T_i$ such that $i \neq j$.
	For every equivalence class $T_i$ such that $i \in J$, we consider the vertices $B_i \subseteq T_i$ that are in distance at most 4 in $G - X$.
	Since the degree of every vertex in $G - X$ is at most $k$ in $G - X$, it must be that $|B_i| \leq k^4$.
	Now, we guess a subset $A_i \subseteq B_i$ containing $v_i$ that we consider to be part of the solution and every vertex of $A_i$ is at distance at most 2 in $G[A_i]$.
	We consider the vertex set $Z_1 = \cup_{i \in [J]} A_i$ that are the set of guessed vertices from all the equivalence classes that are in $Y$-conflict with some other equivalence class.
	Our current partial set is $Y \cup Z_1$.
	
	Now, we look at $I = J \setminus [q]$, the indices corresponding to the equivalence classes $T_i$ that are not in $Y$-conflict with any other equivalence class and consider $W = \cup_{i \in I} T_i$.
	For every equivalence class $T_i$ such that $i \in I$, are yet to decide which vertex to add to $Y \cup Z_1$ and which vertex we are not to add.
	
	We iteratively remove vertices from $W$. If there exist a vertex $w \in W$ that participates in less than $r$ triangles in the graph $G[Y \cup Z_1 \cup W]$, we remove $w$ from $W$, update the triangle counts for each vertex in $W$ (as $W$ has changed) and repeat the process until no such vertex remains in $W$.
	
	
	If $Y \cup Z_1 \cup W$ is a 2-club of size at least $\ell$ and satisfies vertex $r$-triangle property, then we output that the instance is a yes-instance. 
	Otherwise, this guess of $Y \subseteq X$ and the equivalence classes is incorrect, and we move to a different guess.

Note that deletions from $W$ do not destroy the $2$-club property: every vertex of $W$ (and of $Z_1$) has a common neighbour in $Y$ with every other vertex of $Y \cup Z_1 \cup W$. Those witnesses lie in $Y$, and we never delete vertices from $Y$; hence every pair of vertices in the final set $Y \cup Z_1 \cup W$ still has a length-$\leq 2$ path (via the same witness from $Y$), so the $2$-club property is preserved. Any vertex that we delete from $W$ cannot be part of any optimal solution as it does not participate in $r$ triangles. Thus, if we correctly guessed the subset $Y \cup Z_1$, then $Y \cup Z_1 \cup W$ is our optimal solution.

	Now, we argue the running-time of our algorithm.
	There are $2^{|X|}$ many guesses of $X$, and for every such guess $Y$, there are $(n+1)^{2^{|Y|}}$ many guesses of twin-classes with a vertex to be part of the solution.
	Subsequently, for every equivalence class $T_i$ such that $i \in J$, we guess over all subsets of size at most $k^4$ for the choice of $A_i$.
	After that, we process the vertices of the equivalence classes in $T_j$ for all $j \in I$ in polynomial-time.
	Hence, the overall running-time of our algorithm is $2^k n^{\OO(2^k)} 2^{\OO(k^4)} n^{\OO(1)}$-time.
	As ${\hind}(G) = k$, this provides us algorithm with our desired running-time.
\end{proof}



\section{Conclusions}
\label{sec:conclusion}


Our paper has kickstarted a systematic study of the structural parameterizations of {\rVT} {\twoclub} for every $r \geq 1$.
But, for every $s \geq 3$, the FPT status of {\rVT} {\sclub} parameterized by ${\trw}$ is an open question.
It is also unclear if our Lemma \ref{lemma:bipartite-algorithm} can be extended to an XP algorithm when distance to bipartite graph is considered as parameter.
Hence, FPT status of {\rVT} {\twoclub} for distance to bipartite graph also remains open.
Additionally, other interesting graph parameters, such as distance to cluster, distance to co-cluster, degeneracy etc would be interesting future research direction.

%




\begin{thebibliography}{10}

\bibitem{AlmeidaB2019}
Maria~Teresa Almeida and Ra{\'{u}}l Br{\'{a}}s.
\newblock The maximum l-triangle k-club problem: Complexity, properties, and
  algorithms.
\newblock {\em Computers and Operations Research}, 111:258--270, 2019.

\bibitem{Balasundaram09}
Balabhaskar Balasundaram.
\newblock {\em Graph theoretic generalizations of clique: optimization and
  extensions}.
\newblock PhD thesis, Texas A{\&}M University, College Station, {USA}, 2009.

\bibitem{BalasundaramBT05}
Balabhaskar Balasundaram, Sergiy Butenko, and Svyatoslav Trukhanov.
\newblock Novel approaches for analyzing biological networks.
\newblock {\em J. Comb. Optim.}, 10(1):23--39, 2005.

\bibitem{CyganFKLMPPS15}
Marek Cygan, Fedor~V. Fomin, Lukasz Kowalik, Daniel Lokshtanov, D{\'{a}}niel
  Marx, Marcin Pilipczuk, Michal Pilipczuk, and Saket Saurabh.
\newblock {\em Parameterized Algorithms}.
\newblock Springer, 2015.

\bibitem{DiestelBook}
Reinhard Diestel.
\newblock {\em Graph Theory, 4th Edition}, volume 173 of {\em Graduate texts in
  mathematics}.
\newblock Springer, 2012.

\bibitem{Eppstein2015Metric}
David Eppstein.
\newblock Metric dimension parameterized by max leaf number.
\newblock {\em Journal of Graph Algorithms and Applications}, 19(1):313--323,
  2015.

\bibitem{GarvardtKS23}
Jaroslav Garvardt, Christian Komusiewicz, and Frank Sommer.
\newblock The parameterized complexity of s-club with triangle and seed
  constraints.
\newblock {\em Theory Comput. Syst.}, 67(5):1050--1081, 2023.

\bibitem{HartungKN15}
Sepp Hartung, Christian Komusiewicz, and Andr{\'{e}} Nichterlein.
\newblock Parameterized algorithmics and computational experiments for finding
  2-clubs.
\newblock {\em J. Graph Algorithms Appl.}, 19(1):155--190, 2015.

\bibitem{HartungKNS15}
Sepp Hartung, Christian Komusiewicz, Andr{\'{e}} Nichterlein, and Ondrej
  Such{\'{y}}.
\newblock On structural parameterizations for the 2-club problem.
\newblock {\em Discret. Appl. Math.}, 185:79--92, 2015.

\bibitem{Kanawati2014}
Rushed Kanawati.
\newblock Seed-centric approaches for community detection in complex networks.
\newblock In Gabriele Meiselwitz, editor, {\em 6th international conference on
  Social Computing and Social Media}, volume LNCS 8531, pages 197--208, Crete,
  Greece, June 2014. Springer.

\bibitem{Kloks94}
Ton Kloks.
\newblock {\em Treewidth, Computations and Approximations}, volume 842 of {\em
  Lecture Notes in Computer Science}.
\newblock Springer, 1994.

\bibitem{Korhonen21}
Tuukka Korhonen.
\newblock A single-exponential time 2-approximation algorithm for treewidth.
\newblock In {\em 62nd {IEEE} Annual Symposium on Foundations of Computer
  Science, {FOCS} 2021, Denver, CO, USA, February 7-10, 2022}, pages 184--192.
  {IEEE}, 2021.

\bibitem{Mokken16}
Robert~J. Mokken.
\newblock Cliques, clubs, and clans.
\newblock {\em Quality and Quantity.}, 13(2):161--173, 1979.

\bibitem{PajouhB12}
Foad~Mahdavi Pajouh and Balabhaskar Balasundaram.
\newblock On inclusionwise maximal and maximum cardinality k-clubs in graphs.
\newblock {\em Discret. Optim.}, 9(2):84--97, 2012.

\bibitem{PajouhBH16}
Foad~Mahdavi Pajouh, Balabhaskar Balasundaram, and Illya~V. Hicks.
\newblock On the 2-club polytope of graphs.
\newblock {\em Oper. Res.}, 64(6):1466--1481, 2016.

\bibitem{PattilloYB2013}
Jeffrey Pattillo, Nataly Youssef, and Sergiy Butenko.
\newblock On clique relaxation models in network analysis.
\newblock {\em European Journal of Operational Research}, 226(1):9--18, 2013.
\newblock URL:
  \url{https://www.sciencedirect.com/science/article/pii/S0377221712007679},
  \href {https://doi.org/10.1016/j.ejor.2012.10.021}
  {\path{doi:10.1016/j.ejor.2012.10.021}}.

\end{thebibliography}



\end{document}